\documentclass{elsarticle}
\usepackage[T1]{fontenc}

\usepackage{amssymb}
\usepackage{url}

\usepackage{amsmath,amsthm}
\usepackage{mleftright}

\usepackage{algorithm}
\usepackage{algorithmicx}
\usepackage{algpseudocode}
\algnewcommand{\IIf}[1]{\State\algorithmicif\ #1\ \algorithmicthen}
\algnewcommand{\EndIIf}{\unskip\ \algorithmicend\ \algorithmicif}

\makeatletter 
\newenvironment{breakablealgorithm}
{
    \begin{center}
      \refstepcounter{algorithm}
      \hrule height.8pt depth0pt \kern2pt
      \renewcommand{\caption}[2][\relax]{
        {\raggedright\textbf{\fname@algorithm~\thealgorithm} ##2\par}
        \ifx\relax##1\relax 
        \addcontentsline{loa}{algorithm}{\protect\numberline{\thealgorithm}##2}
        \else 
        \addcontentsline{loa}{algorithm}{\protect\numberline{\thealgorithm}##1}
        \fi
        \kern2pt\hrule\kern2pt
      }
    }{
    \kern2pt\hrule\relax
  \end{center}
}
\makeatother

\usepackage{bbm}
\usepackage{xcolor}

\usepackage{enumitem}
\usepackage{hyperref}

\usepackage{url}
\hypersetup{colorlinks}
\usepackage{float}
\usepackage{calc}

\usepackage{graphicx}
\usepackage{subcaption}
\usepackage{longtable}

\newtheorem{lem}{Lemma}
\newtheorem{obs}{Observation}
\newtheorem{thm}{Theorem}
\newtheorem{pro}{Proposition}
\newtheorem{rem}{Remark}
\newtheorem{problem}{Problem}

\ifdef{\weAreInExternalizeTikZMode}{
  \usepackage[a4paper,left=2.5cm,right=2.5cm]{geometry}
  \usetikzlibrary{external}
  \tikzexternalize
  \makeatletter
  \renewcommand{\todo}[2][]{\tikzexternaldisable\@todo[#1]{#2}\tikzexternalenable}
  \makeatother
}
{
  \usepackage[a4paper,left=2.5cm,right=2.5cm]{geometry}
}

\newcommand{\indicator}[2]{\mathbbm{1}_{#1}\mleft(#2\mright)}

\newcommand{\signedDoubleRomanDominationNumber}[1]{\gamma_{\mathrm{sdR}}\mleft(#1\mright)}
\newcommand{\signedDoubleRomanDominationNumberArgfree}{\gamma_{\mathrm{sdR}}}

\newcommand{\signedDoubleKRomanDominationNumber}[2]{\gamma_{\mathrm{sdR},#1}\mleft(#2\mright)}
\newcommand{\signedDoubleKRomanDominationNumberArgfree}[1]{\gamma_{\operatorname{sdR},#1}}
\newcommand{\textSignedDoubleKRomanDomination}{SD$k$RD}

\newcommand{\alphaTotalDominationNumber}[1]{\gamma_{\alpha, \mathrm{t}}\mleft(#1\mright)}

\newcommand{\generalizedPetersen}[2]{P_{#1,#2}}
\newcommand{\flowerSnarks}[1]{\operatorname{FS}_{#1}}

\newcommand{\sizerenormalizedSDRDNumber}[1]{c_{\mathrm{sdR}}\mleft(#1\mright)}

\newcommand{\cumulativeWeight}[2]{w_{#2}(#1)}
\newcommand{\cumulativeWeightFuncargfree}[1]{w(#1)}

\newcommand{\labeldomainSet}{\{\pm{}1,2,3\}}

\newcommand{\nodeLB}{\ell_{\textrm{b}}}
\newcommand{\nodeLBI}{\ell_{\textrm{b,i}}}
\newcommand{\nodeLT}{\ell_{\textrm{t}}}
\newcommand{\nodeLTI}{\ell_{\textrm{t,i}}}

\newcommand{\nodeRB}{r_{\textrm{b}}}
\newcommand{\nodeRBI}{r_{\textrm{b,i}}}
\newcommand{\nodeRT}{r_{\textrm{t}}}
\newcommand{\nodeRTI}{r_{\textrm{t,i}}}

\newcommand\Tspacing{\rule{0pt}{3.5ex}}
\newcommand\Bspacing{\rule[-2.1ex]{0pt}{0pt}}
\newcommand{\TBspacing}{\Tspacing\Bspacing}

\makeatletter
\def\ps@pprintTitle{
  \let\@oddhead\@empty
  \let\@evenhead\@empty
  \let\@oddfoot\@empty
  \let\@evenfoot\@oddfoot
}
\makeatother

\begin{document}
  
  \begin{frontmatter}
    \title{Signed Double Roman Domination on Cubic Graphs}
    \cortext[cor1]{Corresponding author}
    
    \author[inst1]{Enrico Iurlano\corref{cor1}}
    \ead{eiurlano@ac.tuwien.ac.at}
    \author[inst2]{Tatjana Zec}
    \ead{tatjana.zec@pmf.unibl.org}
    \author[inst2]{Marko Djukanovic}
    \ead{marko.djukanovic@pmf.unibl.org}
    \author[inst1]{Günther R.~Raidl}
    \ead{raidl@ac.tuwien.ac.at}
    
    \affiliation[inst1]{organization={Algorithms and Complexity Group, TU Wien},
      addressline={Favoritenstraße 9/1921},
      city={Vienna},
      postcode={1040},
      country={Austria}
    }
    
    \affiliation[inst2]{organization={Faculty of Natural Sciences and Mathematics, University of Banja Luka},
      addressline={Mladena Stojanovi\'ca 2},
      city={\\Banja Luka},
      postcode={78000},
      country={Bosnia and Herzegovina}
    }
    
    \begin{abstract}
      The signed double Roman domination problem is a combinatorial optimization problem on a graph asking to assign a label from $\labeldomainSet$ to each vertex feasibly, such that the total sum of assigned labels is minimized.
      Here feasibility is given whenever (i) vertices labeled $\pm{}1$ have at least one neighbor with label in $\{2,3\}$; (ii) each vertex labeled $-1$ has one $3$-labeled neighbor or at least two $2$-labeled neighbors; and (iii) the sum of labels over the closed neighborhood of any vertex is positive.
      The cumulative weight of an optimal labeling is called signed double Roman domination number (SDRDN).
      In this work, we first consider the problem on general cubic graphs of order $n$ for which we present a sharp $n/2+\Theta(1)$ lower bound for the SDRDN by means of the discharging method.
      Moreover, we derive a new best upper bound.
      Observing that we are often able to minimize the SDRDN over the class of cubic graphs of a fixed order, we then study in this context generalized Petersen graphs for independent interest, for which we propose a constraint programming guided proof.
      We then use these insights to determine the SDRDNs of subcubic $2\times m$ grid graphs, among other results.
    \end{abstract}
    
    
    
  \end{frontmatter}
  
  
  \section{Introduction}
  
  The signed double Roman domination problem (SDRDP) is a natural combination of the classical \emph{signed domination problem}~\cite{dunbar1995signed} and the so-called \emph{double Roman domination problem} \cite{beeler2016double}.
  The latter, in turn, is a variant of the \emph{Roman domination problem} (RDP) \cite{stewart1999defend,cockayne2004roman} well-known from contexts, where it is required to economically distribute resources over a network while still ensuring to have a locally available backup resource; practical application scenarios are, e.g. optimal placement of servers \cite{pagourtzis2002server}, or the reduction of energy consumption in wireless sensor networks \cite{ghaffari2022novel}.
  Originally, the RDP was motivated by a strategy of the Roman emperor Constantine (c.f.~\cite{stewart1999defend}) on how to secure his empire with minimum amount of legions.
  In \cite{hattingh1995algorithmiccomplexity}, it is pointed out that one can use signed domination to model winning strategies for problems where it is required to locally obtain majority votes.

  From the perspective of classical domination, studying cubic graphs has a long tradition.
  In fact, it was already shown in 1980 by Kikuno et al.~\cite{kikuno1980np} that the problem is NP-complete on planar cubic graphs.
  Another influential work was done by Reed~\cite{reed1996paths} in 1996, who derived a sharp upper bound for graphs of minimum vertex degree three; one of his conjectures about the improvability on connected cubic graphs was later falsified and updated in \cite{kostochka2005domination}.
  Apart from the famous dominating set problem, during the last decades, considerable interest has emerged in solving also such more constrained variants of domination problems, in particular their restrictions on specific graph classes: Another important class studied under these aspects is the one of grid graphs for which the dominating set problem \cite{goncalves2011domination}, the $2$-domination problem \cite{rao2019twodomination}, and the RDP \cite{rao2019twodomination} have been solved to optimality.\\
  
  In the following, we consider undirected simple graphs. For such a graph $G=(V,E)$ and a vertex $v\in V$, we denote by $N(v):= \{w\in V\mid vw\in E\}$ the open neighborhood of $v$ and by $N[v] := N(v)\cup\{v\}$ its closure.
  The order of a graph $G$ refers to the cardinality $|V|$ of its set of vertices.
  Graph $G$ is called $d$-regular, if $|N(v)|=d$, for any $v\in V$.
  A \emph{cubic} graph is a $3$-regular graph.
  Given a graph $G=(V,E)$ and a labeling function $f:V\to\mathbb{R}$, for any subset $S\subseteq V$, we define the \emph{cumulative weight} of $f$ restricted to $S$ as $\cumulativeWeight{S}{f} := \sum_{s\in S} f(s)$.
  We also write $\cumulativeWeight{G}{f}$ for $\cumulativeWeight{V}{f}$, and when the function $f$ is clear from the context, we omit $f$ in the subscript.
  Often we directly identify a function $f:V\to\{-1,1,2,3\}$ with its associated preimages $V_i := f^{-1}(\{i\}) = \{v\in V \mid f(v)=i\}$, $i\in\labeldomainSet$.
  We denote $\mathbb{N}=\{0,1,2,\ldots\}$.
  In some definitions, for simplicity, the vertices will be indexed by $\mathbb{Z}_m$, the residue class ring modulo $m$.
  For a set~$A$, by $\indicator{A}{x}$, we refer to its indicator function.
  
  Following \cite{abdollahzadehahangar2019signeddoubledominationINgraphs}, for a given graph $G=(V,E)$, a function $f:V\to\labeldomainSet$ is called \emph{signed double Roman domination function} (SDRDF) on $G$ if the following conditions \eqref{eq:sdrdp-existing-defender-for-minusone-condition}--\eqref{eq:sdrdp-positive-closed-neighborhood-condition} are met.
  \begin{subequations}
    \begin{align}
      & \text{For all }u\in V_{-1},\text{ there exists }v\in N(u)\cap V_3\text{ or there exist distinct }v_1,v_2\in N(u)\cap V_{2}.\label{eq:sdrdp-existing-defender-for-minusone-condition} \\
      & \text{For all }u\in V_{1},\text{ there exists }v\in N(u)\cap (V_2\cup V_3).\label{eq:sdrdp-existing-defender-for-one-condition}                                                      \\
      & \text{For all }u\in V,~\cumulativeWeight{N[u]}{f}\geqslant 1,\text{ i.e., the cumulative weight of }N[u]\text{ is positive.} \label{eq:sdrdp-positive-closed-neighborhood-condition}
    \end{align}
  \end{subequations}
  We call $\signedDoubleRomanDominationNumber{G}:=\min\{\cumulativeWeight{V}{f}\}\mid f\text{ is a SDRDF on }G\}$ \emph{signed double Roman domination number of $G$} (SDRDN).
  Existing vertices $v$, $v_1$ and $v_2$ in \eqref{eq:sdrdp-existing-defender-for-minusone-condition} and \eqref{eq:sdrdp-existing-defender-for-one-condition} are said to \emph{defend} the respective vertex $u$. \\

  A generalization of the SDRDP is the signed double Roman $k$-domination problem (\textSignedDoubleKRomanDomination{}P), originally proposed in~\cite{amjadi2018signed} ($k\in\mathbb{N}\setminus\{0\}$ fixed), requiring the fulfillment of the conditions \eqref{eq:sdrdp-existing-defender-for-minusone-condition}--\eqref{eq:sdrdp-positive-closed-neighborhood-condition} plus the additional restriction $\cumulativeWeight{N[u]}{f}\geqslant k$ for all vertices $u\in V$.
  The minimum weight taken over all labelings satisfying the latter property determines the so-called \emph{\textSignedDoubleKRomanDomination{} number} $\signedDoubleKRomanDominationNumber{k}{G}$.
  
  We introduce notation for special classes of (sub)cubic graphs in what follows:
  On the one hand, for $m\in\mathbb{N}\setminus\{0,1,2\}$ and $k\in\mathbb{Z}_m\setminus\{0\}$, the \emph{generalized Petersen graph} $\generalizedPetersen{m}{k}$ comprises vertex set $\{u_i, v_i \mid i \in \mathbb{Z}_m\}$ and has edge set $\{u_i u_{i+1},v_i v_{i+k},u_i v_{i} \mid i \in\mathbb{Z}_m\}$.
  We refer to the value $k\in\mathbb{Z}_m$ as \emph{shift parameter} and remark that $\generalizedPetersen{m}{1}$ is isomorphic to the \emph{$m$-prism graph}.
  
  On the other hand, we define the $\ell\times m$ \emph{grid graph $G_{\ell,m}$} on the set of vertices $\{0,\ldots,\ell-1\}\times \{0,\ldots,m-1\}\subseteq \mathbb{R}\times\mathbb{R}$, for which two vertices are adjacent if their Euclidean distance equals one \cite{cockayne2004roman}.
  For $\ell=2$ we introduce a briefer notation which identifies $(0,i)\in\mathbb{R}^2$ with the symbol $u_i$ and $(1,i)\in\mathbb{R}^2$ with $v_i$, $i=0,\ldots,m-1$.
  
  Finally, a flower snark $\flowerSnarks{m}$ ($m\geqslant 5$) is a graph with vertex set $V = \{a_i,b_i,c_i,d_i\mid i\in\mathbb{Z}_m\}$ and edge set $E$ formed by the union of the three sets $\{a_ib_i,a_ic_i,a_id_i\mid i\in\mathbb{Z}_m\}$, $\{b_ib_{i+1}\mid i\in \mathbb Z_m\}$, and $\{c_0c_1,c_1c_2,\ldots,c_{m-2}c_{m-1},c_{m-1}d_0,d_0d_1,d_1d_2,\ldots,\allowbreak{}d_{m-2}d_{m-1},c_0d_{m-1}\}$.
  
  These three specific graph classes are visualized in Figure~\ref{fig:introductory-illustartions}.
  
  \begin{figure}[htb!]
    \centering
    \subcaptionbox{Grid graph $G_{2,4}$. \label{fig:grid_graph}}[0.25\textwidth]{
      \centering 
      \includegraphics{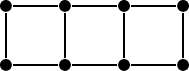}
    }
    \subcaptionbox{Generalized Petersen graph $P_{8,3}$ (Möbius-Kantor graph).
      \label{fig:petersen_p_n_k}}[0.35\textwidth]{
      \centering
      \includegraphics{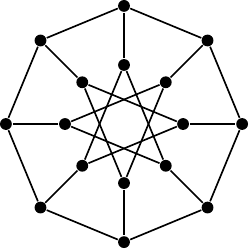}
    }
    \subcaptionbox{Flower snark $\flowerSnarks{5}$. \label{fig:flower-snark-introductory}}[0.38\textwidth]{
      \centering
      \includegraphics{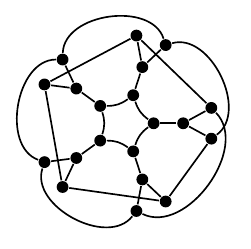}
    }
    \caption{Exemplary graphs for the special graph classes considered in this work.}\label{fig:introductory-illustartions}
  \end{figure}
  
  \medskip
  
  The main contributions of this work are as follows.
  
  \begin{itemize}
    \item A lower bound for $\signedDoubleRomanDominationNumberArgfree$ on cubic graphs twice as high as the so far best known one is derived via the discharging method.
    It turns out to even be optimally sharp, settling the missing case $k=1$ of the collection of optimal lower bounds for the \textSignedDoubleKRomanDomination{}P pointed out in \cite{amjadi2018signed}.
    \item Tight or even optimal bounds on $\signedDoubleRomanDominationNumberArgfree$ are established and proved for
    \begin{itemize}
      \item selected subclasses of generalized Petersen graphs,
      \item $2\times m$ grid graphs,
      \item and flower snarks.
    \end{itemize}
    For some results we design an inductive proof relying on constraint programming~\cite{rossi2006handbook}.
    \item Additionally, best known upper bounds for $\signedDoubleRomanDominationNumberArgfree$ and $\signedDoubleKRomanDominationNumberArgfree{2}$ on (connected) cubic graphs are improved.
  \end{itemize}
  
  \medskip
  
  In the remainder of this introduction, we give an overview of relevant recent results from the literature.
  
  For the SDRDP, it is shown that calculating $\signedDoubleRomanDominationNumberArgfree$ on bipartite as well as on chordal graphs is NP-hard~\cite{abdollahzadehahangar2019signeddoubledominationINgraphs}.
  Moreover, exact values of $\signedDoubleRomanDominationNumberArgfree$ are established for special classes of graphs, including complete graphs, paths, cycles, and complete bipartite graphs.
  In~\cite{abdollahzadehahangar2019signeddoubledominationOFgraphs}, lower bounds for $\signedDoubleRomanDominationNumberArgfree$ are obtained in dependence of the minimum respectively maximum vertex degree; furthermore, it is shown that in the absence of isolated vertices $\signedDoubleRomanDominationNumber{G} \geqslant (19n-24m)/9$, where $n$ and $m$ denote the order of $G$ and the number of edges in $G$, respectively.
  For trees, in \cite{abdollahzadehahangar2019signeddoubledominationINgraphs}, it is shown that $\signedDoubleRomanDominationNumberArgfree \geqslant 4\sqrt{n/3}-n$ and that trees attaining the bound can be characterized.
  Calculating $\signedDoubleRomanDominationNumberArgfree$ on digraphs is addressed in~\cite{amjadi2021signed}.
  
  Results concerning upper bounds for the \textSignedDoubleKRomanDomination{} number $\signedDoubleKRomanDominationNumberArgfree{k}$ on general graphs as well as on specific graph classes such as regular graphs and bipartite graphs are given in \cite{amjadi2018signed}.
  
  More specifically, we are interested in improving the following result.
  \begin{thm}[{\cite[Theorem~3.4]{amjadi2018signed}}]\label{thm:almost-all-sharp-sdkrdp-bounds}
    In the setting of connected cubic graphs\footnote{The lower bound also applies for non-connected cubic graphs \cite[Proposition~2]{amjadi2018signed}.}, the following bounds for $\signedDoubleKRomanDominationNumberArgfree{k}$ apply.
    Moreover, the lower bounds are optimal for $k\in\{2,3,4,5\}$.
    \begin{equation}
      \frac{kn}{4}\leqslant \signedDoubleKRomanDominationNumberArgfree{k}\leqslant \frac{13n}{8}.\label{eq:amjadi-collection-of-k-variant-almost-all-optimal-bounds}
    \end{equation}
  \end{thm}
  In contrast to the trivial worst-case upper bound $\signedDoubleRomanDominationNumberArgfree\leqslant 2n$ on general graphs, this shows that a smaller upper bound can be achieved on cubic graphs.
  In fact, for $k=1$ the latter result just affirms (for connected cubic graphs)
  \begin{equation}
    \frac{n}{4}\leqslant\signedDoubleRomanDominationNumberArgfree\leqslant \frac{13n}{8}.\label{eq:almost-all-sharp-sdkrdp-bounds-for-the-case-k-equal-one}
  \end{equation}
  
  As an auxiliary tool, we will fall back on the following concept from \cite{henning2012alpha}, the so-called \emph{$\alpha$-total domination number $\alphaTotalDominationNumber{G}$}.
  For $0<\alpha<1$, $\alphaTotalDominationNumber{G}$ is defined as the minimum cardinality of an \emph{$\alpha$-total dominating set of $G$}, i.e., a total dominating set $S\subseteq V$ satisfying that any vertex $v\in V\setminus S$ fulfills $|N(v)\cap S| \geqslant \alpha |N(v)|$.
  
  \begin{thm}[{\cite[Theorem 10.b]{henning2012alpha}}]\label{thm:alpha-total-domination-literature-result-for-cubic-graphs}
    Let $G$ be a cubic graph of order $n$.
    For $1/3<\alpha \leqslant 2/3$, we have $n/2\leqslant\alphaTotalDominationNumber{G}< 3n/4$.
  \end{thm}

  \section{Main results}\label{sec:new_results}
  
  We employ $\alpha$-total domination to improve the upper bound in Theorem~\ref{thm:almost-all-sharp-sdkrdp-bounds} (for $k=1$ and $k=2$) by a factor of approximately $0.77$.
  \begin{pro}\label{pro:general-upper-bound-on-cubic-graphs-by-alpha-domination-proof}
    We have $\signedDoubleKRomanDominationNumber{2}{G}<5n/4$ and $\signedDoubleRomanDominationNumber{G}<5n/4$ for cubic graphs $G$ of order $n$.
  \end{pro}
  \begin{proof}
    For $G=(V,E)$, we select a totally dominating subset $S\subseteq V$ such that each vertex $v\in V\setminus S$ has at least two neighbors in $S$, which corresponds to an $\alpha$-total dominating set in $G$ with $\alpha=2/3$.
    Pick the labeling $f$ satisfying $V_2 = S$ and $V_{-1} = V\setminus S$.
    
    We check that the cumulative weight of any closed neighborhood is at least $2$:
    In the neighborhood of any vertex $v\in V_2$, at least one neighbor must be labeled $2$ by total domination.
    Consequently $\cumulativeWeight{N[v]}{f}\geqslant 2$.
    On the other hand, each $v\in V_{-1}$ has at least two neighbors in $V_2$ (by the $\alpha$-domination property), again verifying $\cumulativeWeight{N[v]}{f}\geqslant 2$.
    Adding up all labels, according to Theorem~\ref{thm:alpha-total-domination-literature-result-for-cubic-graphs} we obtain
    \begin{equation*}
      \cumulativeWeight{V}{f} = 2|V_2|-|V_{-1}| < 2\cdot \frac{3n}{4} - \frac{n}{4} = \frac{5n}{4}.
    \end{equation*}
  \end{proof}
  Since we managed to reduce the upper bound \eqref{eq:amjadi-collection-of-k-variant-almost-all-optimal-bounds}, as in \cite{amjadi2018signed}, we pose ourselves the question if $\signedDoubleRomanDominationNumberArgfree\leqslant n$ for connected cubic graphs; see Section~\ref{sec:conclusions} for further thoughts.
  
  Let us add an observation stating that in the setting of \emph{cubic graphs}, formulating that a labeling $f$ is a SDRDF, is expressible in an arithmetic-free manner.
  It will be useful to abbreviate the verification of the SDRDF property in many situations.
  
  \begin{obs}
    Condition \eqref{eq:sdrdp-positive-closed-neighborhood-condition} can be replaced by the following equivalent one.
    \begin{align}
      \text{For all $v\in V$, there are distinct $v_1,v_2 \in N[v]$ such that $-1\not\in\{f(v_1), f(v_2)\}$.}\tag{\ref*{eq:sdrdp-positive-closed-neighborhood-condition}'}\label{eq:replacement-condition}
    \end{align}
    More precisely, it is possible to replace \eqref{eq:sdrdp-existing-defender-for-minusone-condition}--\eqref{eq:sdrdp-positive-closed-neighborhood-condition} by the conjunction of \eqref{eq:sdrdp-existing-defender-for-minusone-condition}--\eqref{eq:sdrdp-existing-defender-for-one-condition} and \eqref{eq:replacement-condition}.
  \end{obs}
  \begin{proof}
    We start by showing that our altered condition implies the original one \eqref{eq:sdrdp-existing-defender-for-minusone-condition}--\eqref{eq:sdrdp-positive-closed-neighborhood-condition}.
    Firstly, if $v\in V_{2}\cup V_3$ and there is at least one further positively labeled vertex in $N(v)$, positivity of $\cumulativeWeight{N[v]}{f}$ ensues.
    Secondly, any $v\in V_{1}$ verifying \eqref{eq:sdrdp-existing-defender-for-one-condition} and  \eqref{eq:replacement-condition} implies the existence of a vertex in $(V_{2}\cup V_3)\cap N[v]$ allowing to conclude $\cumulativeWeight{N[v]}{f}\geqslant 1+2+2\cdot(-1)=1$.
    Thirdly, any $v\in V_{-1}$ with two distinct vertices $v_1,v_2\in N(v)\setminus V_{-1}$ satisfying \eqref{eq:sdrdp-existing-defender-for-minusone-condition} must fulfill $\{f(v_1),f(v_2)\}\in \{ \{2\}, \{3,1\}, \{3,2\}, \{3\}\}$ implying \eqref{eq:sdrdp-positive-closed-neighborhood-condition}.
    
    Now we address the other proof direction by proving its contrapositive:
    Suppose $|N[v]\cap V_{-1}|\geqslant 3$ for some $v\in V$.
    This automatically implies, for some $x\in\labeldomainSet$, that $\cumulativeWeight{N[v]}{f}= -3+x \leqslant 0$.
    We can therefore certify invalidity of \eqref{eq:sdrdp-positive-closed-neighborhood-condition} for the labeling.
  \end{proof}
  We come up with the subsequent lower bound on cubic graphs, which improves upon \eqref{eq:almost-all-sharp-sdkrdp-bounds-for-the-case-k-equal-one} by a factor of two.
  Later, in Remark~\ref{rem:comment-that-bound-indeed-optimal-as-claimed-in-advance}, we will show this lower bound to even be optimal.
  \begin{thm}\label{thm:cubic-graph-perfectly-sharp-bound}
    For any cubic graph $G$ of order $n$ we have
    \begin{equation}
      \signedDoubleRomanDominationNumber{G} \geqslant \begin{cases}
        n/2   & \text{if }n\equiv 0\pmod{4}  \\
        n/2+1 & \text{if }n\equiv 2\pmod{4}.
      \end{cases}\label{eq:lowest-bound-scope-cubic-graphs}
    \end{equation}
  \end{thm}
  \begin{proof}
    First, note that odd values for $n$ in \eqref{eq:lowest-bound-scope-cubic-graphs} are irrelevant, as it is well known that vertex sets of cubic graphs have even cardinality, according to the Handshaking Lemma.
    The proof is divided into two steps.
    
    \medskip
    
    \noindent\emph{Step 1. The lower bound $n/2$ applies.}
    
    Let $f$ be an arbitrary SDRDF on $G$.
    We define the function $g$ as the final product of the following discharging rules \eqref{ite:discharging-rule-zero}--\eqref{ite:discharging-rule-three}, executed one by one in succession; cf.~\cite{shao2018discharging}.
    In these discharging rules, we think of the vertex $v$ as transmitting the charge quantity $1/4$, $3/4$, respectively $5/4$ to each of its specified neighbors.
    \begin{enumerate}[label={(R\arabic*)}, ref=R\arabic*]
      \setcounter{enumi}{-1}
      \item For each $v\in V$, let $g(v) = f(v)$ at the beginning of the procedure. \label{ite:discharging-rule-zero}
      \newlength{\deltalen}
      \setlength{\deltalen}{\widthof{Update $g(v)$}-\widthof{update $g(u)$}}
      \item Update $g(v)$ $\gets g(v) - |N(v)\cap V_{-1}|/4$, for all $v\in V_{1}$, and\newline update $g(u)$\hspace{\deltalen} $\gets g(u) + 1/4$, for all $u\in N(v)\cap V_{-1}$.\label{ite:discharging-rule-one}
      \item Update $g(v)$ $\gets g(v) - 3|N(v)\cap V_{-1}|/4$, for all $v\in V_{2}$, and\newline update $g(u)$\hspace{\deltalen} $\gets g(u) + 3/4$, for all $u\in N(v)\cap V_{-1}$.
      \label{ite:discharging-rule-two}
      \item Update $g(v)$ $\gets g(v) - 5|N(v)\cap V_{-1}|/4$, for all $v\in V_{3}$, and\newline update $g(u)$\hspace{\deltalen} $\gets g(u) + 5/4$, for all $u\in N(v)\cap V_{-1}$.
      \label{ite:discharging-rule-three}
    \end{enumerate}
    
    We note that in this procedure, after any rule application, the equality $\cumulativeWeight{V}{g} = \cumulativeWeight{V}{f}$ is preserved.
    Observe that after the termination of this procedure, we have $g(v)\geqslant 1/2$ for each vertex $v\in V$:
    By cubicity, condition \eqref{eq:replacement-condition} ensures that each $v\not\in V_{-1}$ is adjacent to at most two vertices labeled $-1$ and each $v\in V_{-1}$ is adjacent to at most one vertex labeled $-1$.
    Hence, after application of all the rules \eqref{ite:discharging-rule-zero}--\eqref{ite:discharging-rule-three} on $f$, we obtain the subsequent implications.
    \begin{align}
      v\in V_1                                   & \implies g(v)\geqslant f(v)-2\cdot \frac{1}{4} = \frac{1}{2},\label{eq:case-one-label-is-reduced-by-onehalf}   \\ 
      v\in V_2                                   & \implies g(v)\geqslant f(v)-2\cdot \frac{3}{4} = \frac{1}{2},                                                  \\
      v\in V_3                                   & \implies g(v)\geqslant f(v)-2\cdot \frac{5}{4} = \frac{1}{2},                                                  \\
      v\in V_{-1}\land N(v)\cap V_3=\emptyset    & \implies g(v)\geqslant f(v)+2\cdot\frac{3}{4} = \frac{1}{2},\label{eq:case-minusone-has-no-neighborthree}      \\
      v\in V_{-1}\land N(v)\cap V_3\neq\emptyset & \implies g(v)\geqslant f(v)+\frac{1}{4} + \frac{5}{4} = \frac{1}{2}.\label{eq:case-minusone-has-neighborthree}
    \end{align}
    Bound \eqref{eq:case-minusone-has-no-neighborthree} applies since the implication's premise enforces that $v$ must have at least two neighbors labeled $2$.
    On the other hand, bound \eqref{eq:case-minusone-has-neighborthree} applies because, apart from one $3$-labeled neighbor of $v$ given by the premise, there must be one more neighbor from $V\setminus V_{-1}$ (the minimum value of $g(v)$ is obtained in the situation when this neighbor is labeled $1$, and the remaining third neighbor is labeled $-1$, yielding $g(v) = f(v)+1/4 +5/4 = 1/2$).
    Consequently, at the end of this procedure, we have $g(v)\geqslant 1/2$, for each $v\in V$, implying $\cumulativeWeight{V}{f}= \cumulativeWeight{V}{g} =\sum_{v\in V}g(v)\geqslant |V|/2$.
    
    \medskip 
    
    \noindent\emph{Step 2. The lower bound is refinable for $n\equiv 2\pmod{4}$.}
    
    Let $g: V\to\mathbb{R}$ be the function arising from $f$ via the discharging method in Step 1.
    We make a case distinction.
    
    \emph{Case 1.} There is a vertex $s\in V_1\cup V_2\cup V_3$ having less than two neighbors in $V_{-1}$.
    We show that the bound $n/2$ cannot be attained by $f$:
    In fact,
    \begin{align*}
      \sum_{v\in V_1\cup V_2\cup V_3} g(v) & = g(s) + \sum_{v\in V_1\cup V_2\cup V_3\setminus\{s\}} g(v) \\ & \geqslant g(s) + \frac{|V_1 \cup V_2 \cup V_3|-1}{2}\\&\geqslant \indicator{V_1}{s}(1-\tfrac{1}{4}) + \indicator{V_2}{s}(2-\tfrac{3}{4}) + \indicator{V_3}{s}(3-\tfrac{5}{4}) + \frac{|V_1 \cup V_2 \cup V_3|-1}{2}\\
      &> \frac{|V_1 \cup V_2 \cup V_3|}{2}, 
    \end{align*}
    and therefore $\cumulativeWeight{V}{f} = \sum_{v\in V} g(v) > n/2$.
    
    \emph{Case 2.} Assume all vertices in $V_1\cup V_2\cup V_3$ have two neighbors in $V_{-1}$.
    Let $n=4\ell+2$ where $\ell\in\mathbb{N}\setminus\{0\}$.
    For $v\in V_{-1}$ having three neighbors in $V_1\cup V_2\cup V_3$, in \eqref{eq:case-minusone-has-no-neighborthree} and \eqref{eq:case-minusone-has-neighborthree}, we face even strict majorization $g(v) > \tfrac{1}{2}$.
    Therefore, there exists $\varepsilon>0$ such that we can estimate via \eqref{eq:case-one-label-is-reduced-by-onehalf}--\eqref{eq:case-minusone-has-neighborthree}, 
    \begin{align}
      \sum_{v\in V} g(v) & = \sum_{v\in V_1\cup V_2\cup V_3} g(v) + \sum_{\substack{v\in V_{-1}                                                                  \\ |N(v)\cap (V_1\cup V_2\cup V_3)|=2}} g(v) + \sum_{\substack{v\in V_{-1} \\ |N(v)\cap (V_1\cup V_2\cup V_3)|=3}} g(v) \\
      & \geqslant \tfrac{1}{2} |V_1\cup V_2\cup V_3| + \tfrac{1}{2} \left|\{v\in V_{-1}: |N(v)\cap (V_1\cup V_2\cup V_3)|=2\}\right|\nonumber \\ &\hphantom{\geqslant ++} +(\tfrac{1}{2}+\varepsilon) \left|\{v\in V_{-1}: |N(v)\cap (V_1\cup V_2\cup V_3)|=3\}\right|.\label{eq:epsilon-lower-bound-for-strictness-proving}
    \end{align}
    From \eqref{eq:epsilon-lower-bound-for-strictness-proving} we obtain that whenever $\left|\{v\in V_{-1}: |N(v)\cap(V_1\cup V_2\cup V_3)|=3\}\right|\neq\emptyset$, then we have even more strongly $\cumulativeWeight{V}{f}=\cumulativeWeight{V}{g}=\sum_{v\in V} g(v) > |V|/2$.
    Indeed, in our considered case, this non-emptiness occurs:
    An edge-counting argument applied to the fact that the vertices in $V_1\cup V_2\cup V_3$ have precisely two neighbors in $V_{-1}$ and the fact that each vertex in $V_{-1}$ must have \emph{at least} two neighbors in $V_1\cup V_2 \cup V_3$ shows that $|V_1\cup V_2 \cup V_3| \geqslant |V_{-1}|$.
    The set $V_1\cup V_2 \cup V_3$ must be of even cardinality, as for each of its vertices---apart from the two edges connecting the vertex with $V_{-1}$---the third edge must be incident to a vertex in $V_1\cup V_2 \cup V_3$.
    Moreover, this implies that $|V_1\cup V_2 \cup V_3| > 2\ell +1 > |V_{-1}|$.
    The pigeonhole principle shows that at least one vertex labeled $-1$ must have three neighbors in $V_1\cup V_2\cup V_3$.
  \end{proof}
  
  \begin{rem}\label{rem:comment-that-bound-indeed-optimal-as-claimed-in-advance}
    As we will see, the lower bound \eqref{eq:lowest-bound-scope-cubic-graphs} for cubic graphs is optimally sharp, as, e.g., $\generalizedPetersen{n/2}{3}$ are (connected) cubic graphs attaining the bound.
  \end{rem}
  
  \subsection{Cubic graphs with extremal properties: Generalized Petersen graphs}
  
  Let us start our considerations with the following result.
  
  \begin{thm}\label{thm:sdrdf-number-for-gp-m-k-for-even-m-odd-k}
    We have $\signedDoubleRomanDominationNumber{\generalizedPetersen{m}{k}}=m$ whenever $m\geqslant 4$ is even and $k$ is odd.
  \end{thm}
  \begin{proof}
    Choose the labeling with $V_{-1}=\{u_{2i}, v_{2i} \mid i=0,\ldots,m/2-1\}$ and $V_{2} = V\setminus V_{-1}$.
    Then $\cumulativeWeightFuncargfree{V}=m$, and the SDRDF constraints are met.
    In fact, this function has for each vertex $u\in\{u_{2i}\mid i=0,\ldots,m/2-1\}$ the two $2$-labeled defenders $u_{2i-k}$, $u_{2i+k}$.
    By the same index shift $i\mapsto i\pm{}k$, we see that $v\in\{v_{2i}\mid i=0,\ldots,m/2-1\}$ has two defenders.
    Recalling \eqref{eq:replacement-condition}, the existence of these defenders also guarantees that the vertices $u$ and $v$ have positive cumulative weight on their closed neighborhoods.
    For the vertices $w\in V\setminus V_{-1} = V_2 = \{u_{2i+1},v_{2i+1}\mid i=0,\ldots,m/2-1\}$, the positivity is guaranteed by the fact that $\{u_{2i+1},v_{2i+1}\}\subseteq N[w]\cap V_2$.
    
    Finally, as the weight of the constructed SDRDF coincides with the lower bound of the previous Theorem~\ref{thm:cubic-graph-perfectly-sharp-bound}, the SDRDF is optimal.
  \end{proof}

  \begin{thm}\label{thm:generalized-petersen-graph-case-k-equal-three}
    For the generalized Petersen graph $P_{m,3}$, $m\geqslant 8$, we have
    \begin{equation}
      \signedDoubleRomanDominationNumber{P_{m,3}} = \begin{cases}
        m   & \text{if } m\equiv 0 \pmod{2}, \\
        m+1 & \text{else.}                   \\
      \end{cases}\label{eq:generalized-petersen-graph-case-k-equal-three}
    \end{equation}
  \end{thm}
  \begin{proof}
    For even $m$, optimal constructions proving \eqref{eq:generalized-petersen-graph-case-k-equal-three} have already been found, cf. Theorem~\ref{thm:sdrdf-number-for-gp-m-k-for-even-m-odd-k} for $k=3$.
    To show that the right-hand side of \eqref{eq:generalized-petersen-graph-case-k-equal-three} is an \emph{upper bound} for $\signedDoubleRomanDominationNumber{\generalizedPetersen{m}{3}}$ for odd $m$, we distinguish two cases, both constructing a particular SDRDF on $\generalizedPetersen{m}{3}$; in Figure~\ref{fig:generalized-petersen-shiftlen-eq-three-thirteen-and-nineteen} supportive visualizations of the underlying scheme for both are given.
    
    \emph{Case 1.} $m\equiv 1 \pmod{4}$.
    
    Let $f$ be the labeling with $V_2 = \{\allowbreak{}u_{4i},\allowbreak{}u_{4i+1},\allowbreak{}v_{4i+2},\allowbreak{}v_{4i+3}\mid i = 0,\ldots,\frac{m-9}{4}\} \cup \{\allowbreak{}u_{m-5},\allowbreak{}u_{m-4},\allowbreak{}u_{m-2},\allowbreak{}v_{m-2}\}$, 
    $V_1 = \{v_{m-3},v_{m-1}\}$, and 
    $V_{-1} =V\setminus(V_2\cup V_1) = \{\allowbreak{}u_{4i+2},\allowbreak{}u_{4i+3},\allowbreak{}v_{4i},\allowbreak{}v_{4i+1}\mid i = 0,\ldots,\frac{m-9}{4}\}\cup\{\allowbreak{}u_{m-3},\allowbreak{}u_{m-1},\allowbreak{}v_{m-5},\allowbreak{}v_{m-4}\}$.
    The satisfaction of all SDRDF constraints by $f$ is argued in Table~\ref{tab:Pm3_case1} in the appendix.
    This implies $\signedDoubleRomanDominationNumber{P_{m,3}}\leqslant \cumulativeWeight{P_{m,3}}{f} = 2|V_2|+|V_1|-|V_{-1}| = 2(m-1)+2-(m-1)=m+1$.
    
    \emph{Case 2.} $m\equiv 3 \pmod{4}$.
    We construct a labeling $f$ satisfying 
    $V_3 = \{v_{m-3}\}$, $V_2 = \{\allowbreak{} u_{4i+2}, \allowbreak{} u_{4i+3}, \allowbreak{} v_{4i},\allowbreak{} v_{4i+1}\mid i = 0,\ldots,\frac{m-15}{4}\}\cup$ $\{u_{m-9},\allowbreak{}u_{m-7},\allowbreak{}u_{m-5},\allowbreak{}u_{m-1},\allowbreak{}v_{m-11},\allowbreak{}v_{m-10},\allowbreak{}v_{m-5},\allowbreak{}v_{m-4}\}$, $V_1 = \{\allowbreak{}u_{m-2},\allowbreak{}v_{m-9}, \allowbreak{}v_{m-7}\}$, and $V_{-1} =V\setminus(V_2\cup V_1) = \{\allowbreak{}u_{4i},\allowbreak{}u_{4i+1},\allowbreak{}v_{4i+2},\allowbreak{}v_{4i+3}\mid i = 0,\ldots,\frac{m-15}{4}\}\cup\{\allowbreak{}u_{m-11},\allowbreak{}u_{m-10},\allowbreak{}u_{m-8},\allowbreak{}u_{m-6},\allowbreak{}u_{m-4},\allowbreak{}u_{m-3},\allowbreak{}v_{m-8},\allowbreak{}v_{m-6},\allowbreak{}v_{m-2},\allowbreak{}v_{m-1}\}$.
    We check that $f$ is a SDRDF in Table~\ref{tab:Pm3_case2} in the appendix.
    Therefore, we conclude 
    $\signedDoubleRomanDominationNumber{P_{m,3}}\leqslant\cumulativeWeight{P_{m,3}}{f} =3|V_3|+ 2|V_2|+|V_1|-|V_{-1}| =3+ 2(m-3)+3-(m-1)=m+1$.
    
    Finally, it remains to show that the right-hand side of \eqref{eq:generalized-petersen-graph-case-k-equal-three} is also a \emph{lower bound} for $\signedDoubleRomanDominationNumber{\generalizedPetersen{m}{3}}$ when $m$ is odd.
    However, this follows directly from Theorem~\ref{thm:cubic-graph-perfectly-sharp-bound} and concludes our proof.
  \end{proof}
  \begin{figure}[t]
    \centering
    \begin{minipage}{0.87\textwidth}
      \includegraphics{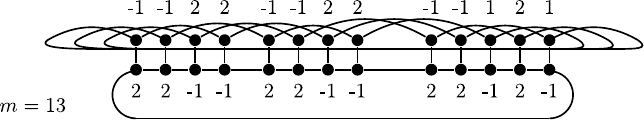}
      
      \medskip
      
      \includegraphics{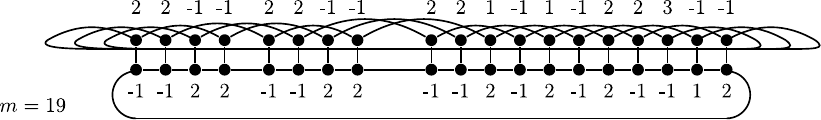}
    \end{minipage}
    \caption{Optimal SDRDFs for $\generalizedPetersen{m}{3}$ when $m=4\ell+1$ respectively $m=4\ell+3$.
      In both cases, a label pattern of width $4$ is periodically repeated $\ell-1$ respectively $\ell-2$ times to finally be flanked by a termination pattern of width $5$ respectively $11$.
      The labeling is exemplarily illustrated for $m=13$ respectively $m=19$.\label{fig:generalized-petersen-shiftlen-eq-three-thirteen-and-nineteen}}
  \end{figure}
  
  In the following, we point out that the graph $\generalizedPetersen{m}{1}$---with the exception of $m\equiv 1 \pmod{4}$---attains the lower bound in \eqref{eq:lowest-bound-scope-cubic-graphs}, too.
  For tackling the aforementioned exceptional case, we state in the following two technical results as Lemma~\ref{lem:constraing-programming-as-mathematical-lemma-formulation} and Lemma~\ref{lem:lower-bound-generalized-petersen-for-the-case-congruent-one-modulo-four}.
  These results incorporate an approach to determine $\signedDoubleRomanDominationNumberArgfree$ for a sufficiently structured rotationally symmetric graph.
  The method relies on a computer-aided exhaustive search for optima on fixed small subgraphs.
  It seems applicable to other domination-like problems, too.

  \begin{figure}[H]
    \begin{subfigure}[]{0.55\textwidth}
      \centering
      \includegraphics{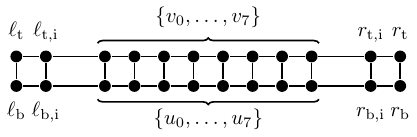}
      \caption{Adjacency of the graph $G$.}\label{fig:two-by-twelve-grid-graph}
    \end{subfigure}
    \hfill 
    \begin{subfigure}[]{0.4\textwidth}
      \centering
      \includegraphics{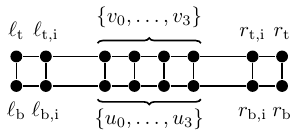}
      \caption{Adjacency of the graph $G'$.}\label{fig:two-by-eight-grid-graph-after-cutoff}
    \end{subfigure}
    \caption{The graph $G'$ in \eqref{fig:two-by-eight-grid-graph-after-cutoff} is the result of deleting four of the vertical edges from $G$ in \eqref{fig:two-by-twelve-grid-graph} and successively performing eight edge contractions.}\label{fig:comparison-twelve-vs-eight-after-deletion-of-eight-vertices}
  \end{figure}
  
  \begin{lem}\label{lem:constraing-programming-as-mathematical-lemma-formulation}
    We consider vertex sets $L := \{\nodeLB,\allowbreak{}\nodeLBI,\allowbreak{}\nodeLT,\allowbreak{}\nodeLTI\}$, $R := \{\nodeRB,\allowbreak{}\nodeRBI,\allowbreak{}\nodeRT,\allowbreak{}\nodeRTI\}$, $C := \{u_i,v_i \mid i=0,\ldots,7\}$, and $C' := \{u_i,v_i \mid i=0,\ldots,3\}$.
    Let $G$ and $G'$ be the grid graphs having vertex sets $V := L \cup C \cup R$ and $V' := L\cup C' \cup R$, respectively, and edges as depicted in Figures~\ref{fig:two-by-twelve-grid-graph} and \ref{fig:two-by-eight-grid-graph-after-cutoff}.
    Let $f:V\to\labeldomainSet$, respectively $f': V'\to\labeldomainSet$ satisfy the SDRDP constraints \eqref{eq:sdrdp-existing-defender-for-minusone-condition}--\eqref{eq:sdrdp-positive-closed-neighborhood-condition} in all vertices except possibly for those in $\{\nodeLT, \nodeLB, \nodeRT, \nodeRB\}$.
    Moreover, let us assume that $f$ attains minimal cumulative weight on $C$ and $f'$ attains minimal cumulative weight on $C'$.\footnote{I.e., $f$ and $f'$ can both not be improved by updating their values just on $C$ and $C'$, respectively.}
    Then, the following properties hold.
    
    \begin{enumerate}[label={(\roman*)}]
      \item For $k\leqslant 5$, $\cumulativeWeight{C}{f} \neq k$.\label{ite:no-block-has-weight-less-than-six}
      \item For $k\in\{6,7,9\}$, whenever $\cumulativeWeight{C}{f} = k$, then $\cumulativeWeight{C'}{f'}=k-4$.\label{ite:replaceability-by-shorter-cutoff}
    \end{enumerate}
  \end{lem}
  
  \begin{proof}
    Exhaustively, per given parameter choice $d=(\nodeLB,\allowbreak\nodeLBI,\allowbreak \nodeLT,\allowbreak\nodeLTI,\allowbreak\nodeRB,\allowbreak \nodeRBI,\allowbreak\nodeRT,\allowbreak\nodeRTI)\in\{\pm{}1,\allowbreak \allowbreak 2,3\}^8$, i.e., by fixing the labels on the delimiting vertices in $L\cup R$, we can determine a SDRDF being minimal with respect to the cumulative weight restricted to $C$ (respectively to $C'$).
    
    Algorithm~\ref{alg:exhaustively-comparing-optima} in the appendix explains how we carried this out computationally.
    After symmetry breaking (see Remark~\ref{rem:symmetry-breaking}), the algorithm exhaustively examines several cases, ultimately showing that the smallest attainable optimal weight is $6$, which proves claim \ref{ite:no-block-has-weight-less-than-six}.
    Furthermore, \ref{ite:replaceability-by-shorter-cutoff} is valid, as we observe that all hereby obtained minima over $C$ attaining the value $k\in\{6,7,9\}$ are accompanied by a respective minimum of $k-4$ on the smaller center $C'$ in $G'$ with the same delimiting constellation $d$.
  \end{proof}
  
  \begin{rem}[Symmetry breaking]\label{rem:symmetry-breaking}
    We employ vertical and horizontal flipping and point reflection through the center, i.e., a labeling for $\begin{bmatrix}
      \nodeLT & \nodeLTI & \nodeRTI & \nodeRT \\
      \nodeLB & \nodeLBI & \nodeRBI & \nodeRB
    \end{bmatrix}$ is oftentimes represented by a respective labeling for $\begin{bmatrix}
      \nodeLB & \nodeLBI & \nodeRBI & \nodeRB \\
      \nodeLT & \nodeLTI & \nodeRTI & \nodeRT
    \end{bmatrix}$, $\begin{bmatrix}
      \nodeRTI & \nodeRT & \nodeLT & \nodeLTI \\
      \nodeRBI & \nodeRB & \nodeLB & \nodeLBI
    \end{bmatrix}$, or $\begin{bmatrix}
      \nodeRB & \nodeRBI & \nodeLBI & \nodeLB \\
      \nodeRT & \nodeRTI & \nodeLTI & \nodeLT
    \end{bmatrix}$.
    Instead of the $4^8=65536$ constellations, it is herewith sufficient to fall back to only a fraction of them, which, after removal of the constellations placing more than two ($-1$)-labels inside $\langle \nodeLT, \nodeLTI, \nodeLB, \nodeLBI\rangle$ or inside $\langle \nodeRTI, \nodeRT, \nodeRBI, \nodeRB\rangle$ (hence violating \eqref{eq:sdrdp-positive-closed-neighborhood-condition}), contains $14940$ cases.
    To keep the argument conceptually simple, we did not eliminate further parameter constellations, which a priori might indicate non-optimality.
  \end{rem}
  
  Given a fixed $\generalizedPetersen{m}{1}$, $m\geqslant 13$ with an optimal SDRDF function $f$ defined on it, we say that a $2\times 12$ subblock of $\generalizedPetersen{m}{1}$, i.e., a subset of vertices $\{v_{i+j},u_{i+j}\mid j=0,\ldots,11\}$ for some $i\in\mathbb{Z}_m$, has the \emph{quality-transferring property w.r.t.\ $f$}, if the vertices $\{\nodeLB,\allowbreak\nodeLBI,\allowbreak \nodeLT,\allowbreak\nodeLTI,\allowbreak\nodeRB,\allowbreak \nodeRBI,\allowbreak\nodeRT,\allowbreak\nodeRTI\}\cup\{v_0,u_0,\ldots,v_3,u_3\}$ of the graph $G'$ in Figure~\ref{fig:two-by-eight-grid-graph-after-cutoff} can be labeled by a function $\tilde{f}$ in such a way that\footnote{Entry-wise equality of $2\times 4$ arrays is meant in \eqref{eq:quality-transferring-boundary-coincides}.}
  \begin{align}
    \begin{bmatrix}
      \tilde{f}(\nodeLT) & \tilde{f}(\nodeLTI) & \tilde{f}(\nodeRTI) & \tilde{f}(\nodeRT) \\
      \tilde{f}(\nodeLB) & \tilde{f}(\nodeLBI) & \tilde{f}(\nodeRBI) & \tilde{f}(\nodeRB)
    \end{bmatrix} = \begin{bmatrix}
      f(v_{i}) & f(v_{i+1}) & f(v_{i+10}) & f(v_{i+11}) \\
      f(u_{i}) & f(u_{i+1}) & f(u_{i+10}) & f(u_{i+11})
    \end{bmatrix},\label{eq:quality-transferring-boundary-coincides}
    \\
    \cumulativeWeight{\{u_0,\ldots,u_3\}\cup \{v_0,\ldots,v_3\}}{\tilde{f}} \leqslant  \cumulativeWeight{\{u_{i+2},\ldots,u_{i+9}\}\cup \{v_{i+2},\ldots,v_{i+9}\}}{f}-4,\label{eq:quality-transferring-after-label-updates-on-smaller-graph-quality-is-preserved}
    \\
    \text{and }\tilde{f}\text{ satisfies \eqref{eq:sdrdp-existing-defender-for-minusone-condition}--\eqref{eq:sdrdp-positive-closed-neighborhood-condition}}\text{ on all vertices not contained in }\{\nodeLT,\nodeLB,\nodeRT,\nodeRB\}.
  \end{align}
  We say that $f$ on $\generalizedPetersen{m}{1}$ is \emph{quality-transferring} if there exists at least one $2\times 12$ subblock having the quality-transferring property w.r.t.\ $f$.
  
  \begin{lem}\label{lem:lower-bound-generalized-petersen-for-the-case-congruent-one-modulo-four}
    Let $m>1$ and $m\equiv 1 \pmod{4}$.
    Then $\signedDoubleRomanDominationNumber{\generalizedPetersen{m}{1}} = m+2$.
  \end{lem}
  
  \begin{proof}
    First, note that the labeling given in Figure~\ref{fig:petersen-graph-case-four-ell-plus-one-valid-labeling-illustration} has cumulative weight $m+2$, further implying for each $m>1$ with $m\equiv 1\pmod{4}$ that
    \begin{equation} \label{eq:sdsdn_lp_m1_eq_m_plus_2}
      \signedDoubleRomanDominationNumber{P_{m,1}}\leqslant m+2.
    \end{equation}

    For $m\leqslant 13$, i.e., for $m\in\{5,9,13\}$, $\signedDoubleRomanDominationNumber{P_{m,1}} = m+2$ follows by exhaustion.
    By complete induction, we settle the case for $m>12\wedge m\equiv 1\pmod{4}$.
    The base case $m=13$ has already been verified.
    As induction hypothesis now assume that for each $\tilde{m}\in\{13,17,\ldots,m\}$, $m\equiv 1\pmod{4}$, the claim holds.
    In the inductive step, we prove that the claim is valid also for $m+4$.
    
    Let $f:V=\{u_i,v_i:i=0,\ldots,m+4-1\}\to\labeldomainSet$ be a minimum weight SDRDF for $\generalizedPetersen{m+4}{1}$.
    We know from \eqref{eq:sdsdn_lp_m1_eq_m_plus_2} that its weight does not exceed $m+4+2$.
    Seeking a contradiction, assume
    \begin{equation}
      \cumulativeWeight{\generalizedPetersen{m+4}{1}}{f} = \cumulativeWeight{\{u_i,v_i\mid i=0,\ldots,m+4-1\}}{f} < m+4+2.\label{eq:assumption-seeking-contradiction-in-inductive-step}
    \end{equation}

    This assumption enforces that no $2\times 12$ subblock can have the quality-transferring property:
    If a subblock, say w.l.o.g. $\{u_{-2},v_{-2},\ldots,u_9,v_9\}$, has this property, we can argue as follows:
    Let $\widetilde{\generalizedPetersen{m}{1}}$ be the graph resulting from $\generalizedPetersen{m+4}{1}$ after deleting vertices $\{u_4,v_4,\ldots,u_7,v_7\}$ and adding the two edges $u_3u_8$, $v_3v_8$.
    Clearly, this graph is isomorphic to $\generalizedPetersen{m}{1}$.
    By the quality-transferring property \eqref{eq:quality-transferring-boundary-coincides}, there exists a function $\tilde{f}$ through which we can define a SDRDF $g$ on $\widetilde{\generalizedPetersen{m}{1}}$ via
    \begin{equation}
      g(z):=\begin{cases}
        \tilde{f}(z) & \text{if }z\in\{v_0,u_0,\ldots,v_3,u_3\} \\
        f(z)         & \text{otherwise}.
      \end{cases}
    \end{equation}
    We conclude that 
    \begin{align}
      \signedDoubleRomanDominationNumber{\generalizedPetersen{m}{1}} \leqslant \cumulativeWeight{\widetilde{\generalizedPetersen{m}{1}}}{g} & = \cumulativeWeight{\generalizedPetersen{m+4}{1}}{f}-\cumulativeWeight{\{u_0,v_0,\ldots,u_7,v_7\}}{f} + \cumulativeWeight{\{u_0,v_0,\ldots,u_3,v_3\}}{\tilde{f}} \\
      & \leqslant \cumulativeWeight{\generalizedPetersen{m+4}{1}}{f} -4 \label{eq:inductive-proof-p-m-one-mod-four-employing-upper-bound-from-quality-transferring}      \\
      & < m+2, \label{eq:inductive-proof-p-m-one-mod-four-employing-contradictory-assumption}
    \end{align}
    where we apply \eqref{eq:quality-transferring-after-label-updates-on-smaller-graph-quality-is-preserved} in step \eqref{eq:inductive-proof-p-m-one-mod-four-employing-upper-bound-from-quality-transferring} and \eqref{eq:assumption-seeking-contradiction-in-inductive-step} in step \eqref{eq:inductive-proof-p-m-one-mod-four-employing-contradictory-assumption}.
    Thus, we obtain a contradiction to our assumption $\signedDoubleRomanDominationNumber{\generalizedPetersen{m}{1}}=m+2$ from the inductive step, so that necessarily
    \begin{equation}
      f\in \{h\mid h:V\to\labeldomainSet\text{ and }h\text{ on }\generalizedPetersen{m+4}{1}\text{ is not quality-transferring}\}.\label{eq:f-is-not-quality-transferring-summarized-conclusion}
    \end{equation}

    By Lemma~\ref{lem:constraing-programming-as-mathematical-lemma-formulation}, we would face for each choice of $i\in\mathbb{Z}_{m+4}$, for each label constellation for $\{u_{i+j},v_{i+j}\mid j\in\{-2,-1,8,9\}\}$ and any labeling of the $2\times 8$ subblock $M_i :=\{u_{i+j},v_{i+j}\mid j=0,\ldots,7\}$ of cumulative weight $k\in\{6,7,9\}$, the quality-transferring property.
    It is therefore impossible, that any $2\times 8$ subblock of $P_{m+4}$ attains the cumulative weight $6$, $7$, or $9$.
    In particular, we have shown that necessarily
    \begin{equation}
      \cumulativeWeight{M_i}{f}\geqslant 8,\text{ for all }i\in\mathbb{Z}_{m+4}.\label{eq:subblock-length-eight-has-weight-at-least-eight}
    \end{equation}

    By Theorem~\ref{thm:cubic-graph-perfectly-sharp-bound} we know $\signedDoubleRomanDominationNumber{\generalizedPetersen{m+4}{1}}\geqslant m+4+1$.
    Hence, there must exist an index $i'$ such that $\cumulativeWeight{M_{i'}}{f}\geqslant 9$---otherwise, we would have $\cumulativeWeight{M_i}{f}=8$ for all $i$, implying $\cumulativeWeight{V}{f} = \sum_{i\in \mathbb{Z}_{m+4}} \cumulativeWeight{M_i}{f}/8 = 8(m+4)/8 = m+4$ and contradicting Theorem~\ref{thm:cubic-graph-perfectly-sharp-bound}.
    However, for $i'$ we even must have $\cumulativeWeight{M_{i'}}{f}\geqslant 10$ according to our previously observed impossibility to attain weight $9$.
    
    To conclude that $\cumulativeWeight{\generalizedPetersen{m+4}{1}}{f}< m+4+2$ always leads to a contradiction, we distinguish two cases.
    
    \emph{Case 1.} Suppose $m+4=8\ell+5$, $\ell\in\mathbb{N}$.
    
    Observation~\ref{obs:thoughts-with-sufficently-many-flanking-columns}~\ref{ite:cut-off-intricate-case-residue-mod-eight-is-equal-to-one-residue-mod-eight-is-equal-to-five} tells us that either $f$ on $\generalizedPetersen{m+4}{1}$ has the quality-transferring property (immediate contradiction to \eqref{eq:f-is-not-quality-transferring-summarized-conclusion}) or there exists a suitable index $i(5)\in\mathbb{Z}_{m+4}$ for which $A:=\{u_{i(5)},\allowbreak{}v_{i(5)},\ldots,u_{i(5)+12},\allowbreak{}v_{i(5)+12}\}$ induces a $2\times 13$ subblock of cumulative weight not smaller than $15$ leading to a lower bound exceeding the upper bound in \eqref{eq:assumption-seeking-contradiction-in-inductive-step}, as can be seen via the following argument:  
    Partition the vertices of $V\setminus A$ into $\ell -1$ subblocks of dimensions $2\times 8$, and apply \eqref{eq:subblock-length-eight-has-weight-at-least-eight} on them.
    Then, $\cumulativeWeight{V}{f}=\cumulativeWeight{V\setminus A}{f}+\cumulativeWeight{A}{f}$ which can be bounded from below by $8(\ell-1) + 15 = 8\ell + 5 +2 = m+4+2$ and contradicts \eqref{eq:assumption-seeking-contradiction-in-inductive-step}.
    
    \emph{Case 2.} Suppose $m+4=8\ell+1$, $\ell\in\mathbb{N}$.
    
    Observation~\ref{obs:thoughts-with-sufficently-many-flanking-columns}~\ref{ite:cut-off-intricate-case-residue-mod-eight-is-equal-to-one-residue-mod-eight-is-equal-to-one} guarantees that either $f$ on $\generalizedPetersen{m+4}{1}$ has the quality-transferring property (immediate contradiction to \eqref{eq:f-is-not-quality-transferring-summarized-conclusion}) or there exists a suitable index $i(1)\in\mathbb{Z}_{m+4}$ for which $\{u_{i(1)},v_{i(1)},\ldots,u_{i(1)+8},v_{i(1)+8}\}$ induces a $2\times 9$ subblock of cumulative weight not smaller than $11$ leading to a lower bound exceeding the upper bound in \eqref{eq:assumption-seeking-contradiction-in-inductive-step}, as can be seen via the following argument:
    Similarly as before we can estimate $\cumulativeWeight{\generalizedPetersen{m+4}{1}}{f}\geqslant 8(\ell-1)+11=8\ell+1+2=m+4+2$, yielding again a contradiction to \eqref{eq:assumption-seeking-contradiction-in-inductive-step}.
  \end{proof}
  
  \begin{thm}\label{thm:generalized-petersen-graph-case-k-equal-one}
    For the generalized Petersen graph $\generalizedPetersen{m}{1}$, $m\geqslant 3$, we have
    \begin{equation}
      \signedDoubleRomanDominationNumber{\generalizedPetersen{m}{1}} = \begin{cases}
        m   & \text{if }m\equiv 0 \pmod{2}   \\
        m+1 & \text{if } m\equiv 3 \pmod{4}  \\
        m+2 & \text{if } m\equiv 1 \pmod{4}. \\
      \end{cases}\label{eq:generalized-petersen-graph-case-k-equal-one}
    \end{equation}
  \end{thm}

  \begin{proof} For even $m$, $\signedDoubleRomanDominationNumber{\generalizedPetersen{m}{1}}=m$ follows directly from Theorem~\ref{thm:sdrdf-number-for-gp-m-k-for-even-m-odd-k} for $k=1$.
    For $m=4\ell+1$ the claim has been shown in Lemma~\ref{lem:lower-bound-generalized-petersen-for-the-case-congruent-one-modulo-four}.
    The \emph{upper bound} for the case $m=4\ell+3$ is given in Figure~\ref{fig:petersen-graph-case-two-ell-valid-labeling-illustration-residue-equal-three}.
    
    \begin{figure}[htb!]
      \centering
      \subcaptionbox{Scheme for $\generalizedPetersen{m}{1}$ when $m=4\ell+1$.
        The graph is depicted for $(m,\ell)=(9,2)$, and its SDRDF weight is $m+2=11$.
        For general $m$ the labeling satisfies $V_{-1}=\{u_i,v_i\mid i=1,3,5\ldots,m-4\}\cup\{u_{m-1},v_{m-2}\}$, $V_{1}=\{v_{m-1}\}$, $V_{2}=\{0,2,4,\ldots,m-5\}\cup \{u_{m-3},v_{m-3}, u_{m-2}\}$, and $V_3=\emptyset$.\label{fig:petersen-graph-case-four-ell-plus-one-valid-labeling-illustration}}
      [0.45\textwidth]{
        \includegraphics{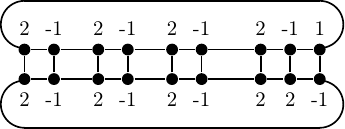}
      }
      \hfill
      \subcaptionbox{Scheme for $\generalizedPetersen{m}{1}$ when $m=4\ell+3$.
        The graph is illustrated for $(m,\ell) = (11,2)$, and its SDRDF weight is $m+1=12$.
        For general $m$ the labeling satisfies $V_{-1}=\{u_{4t},u_{4t+1},v_{4t+2},v_{4t+3}\mid t=0,\ldots,\ell-1\}\cup\{u_{m-3},v_{m-1}\}$, $V_{1}=\{u_{m-2},v_{m-2}\}$, $V_{2}=\{u_{4t+2},u_{4t+3},v_{4t},v_{4t+1}\mid t=0,\ldots,\ell-1\}\cup\{u_{m-3},v_{m-1}\}\cup\{u_{m-1},v_{m-3}\}$, and $V_3=\emptyset$. \label{fig:petersen-graph-case-two-ell-valid-labeling-illustration-residue-equal-three}}[0.50\textwidth]{
        \centering
        \includegraphics{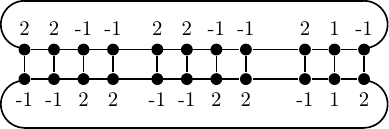}
      }
      
      \caption{Schemes for optimal labelings given in Theorem~\ref{thm:generalized-petersen-graph-case-k-equal-one} for the graph $\generalizedPetersen{m}{1}$.}\label{fig:already-easily-settled-four-ell-plus-odd-generalized-petersen-graph-case-illustrated}
    \end{figure}
    For the \emph{lower bound} for $\signedDoubleRomanDominationNumber{\generalizedPetersen{4\ell+3}{1}}$, we apply Theorem~\ref{thm:cubic-graph-perfectly-sharp-bound} to $n=8\ell +6$ (the count of vertices in $\generalizedPetersen{m}{1}$) and conclude $\signedDoubleRomanDominationNumber{\generalizedPetersen{4\ell+3}{1}} \geqslant n/2+1 = 4\ell + 4 = m+1$.
  \end{proof}
  
  \subsection{Consequences for the grid graph $G_{2,m}$}
  As a byproduct of the results on cubic graphs, particularly on $\generalizedPetersen{m}{1}$, we obtain the following result about optimal SDRDFs on $2\times m$ grid graphs.
  
  \begin{figure}[t]
    \centering
    \begin{subfigure}[]{0.38\textwidth}
      \centering
      \includegraphics{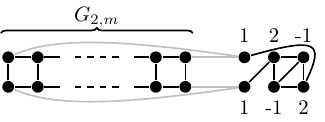}
      \caption{Extending $G_{2,m}$ to a cubic graph with six additional labeled vertices.
        When $m\equiv 0\pmod{2}$, then the extended graph possesses $2m+6\equiv 2\pmod{4}$ vertices.}\label{fig:make-a-single-grid-graph-cubic-by-just-adding-three-vertices}
    \end{subfigure}
    \hfill
    \begin{subfigure}[]{0.55\textwidth}
      \centering
      \includegraphics{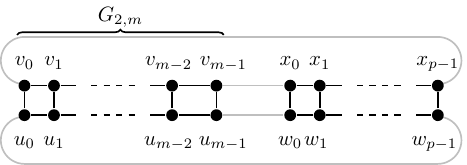}
      \caption{Transforming $G_{2,m}$ into $\generalizedPetersen{m+p}{1}$ by adding suitably connected fresh vertices $x_i$,$w_i$, $i=0,\ldots,p-1$.}\label{fig:add-eight-nodes-to-make-prims-graph-for-bound-proof}
    \end{subfigure}
    \caption{Extending $G_{2,m}$ to a cubic graph via different constructions.}
  \end{figure}

  \begin{thm}
    For $m\geqslant 5$, we have
    \begin{equation}
      \signedDoubleRomanDominationNumber{G_{2,m}}= \begin{cases}
        m+1 & \text{if } m\equiv 1\pmod{4} \\
        m   & \text{otherwise.}
      \end{cases}\label{eq:signed-double-roman-domination-number-formula-for-two-by-m-grid-graphs}
    \end{equation}
  \end{thm}
  \begin{proof}
    For values $m=1,\ldots,13$, the sequence of respective $\signedDoubleRomanDominationNumber{G_{2,m}}$-values can be calculated by exhaustion and corresponds to $\langle 2,4,2,5,6,6,7,8,10,10,11,12,14\rangle$.
    This confirms \eqref{eq:signed-double-roman-domination-number-formula-for-two-by-m-grid-graphs} for $5\leqslant m\leqslant 13$.
    For higher values of $m$, the fact that the right-hand side of \eqref{eq:signed-double-roman-domination-number-formula-for-two-by-m-grid-graphs} majorizes $\signedDoubleRomanDominationNumber{G_{2,m}}$ can be read off the labeling schemata given in Figure~\ref{fig:grid-two-times-m-case-congruences-three-one-mod-four} in the appendix.
    In all four cases, it is easy to recognize that the respective labelings give valid SDRDFs; thus, the respective upper bounds apply.
    
    To show the optimality of the derived upper bounds, we use the subsequent principle, which extends the graph $G_{2,m}$ to a cubic graph.
    For even $m$, the argumentation is less subtle and is better suited to understand the principle.

    Let $m\equiv 0\pmod{2}$.
    Starting from an optimal SDRDF labeled graph $G_{2,m}$, counting $2m\equiv 0\pmod{4}$ vertices, we construct a SDRDF labeled cubic graph $\tilde{G}=(\tilde{V},\tilde{E})$ with six additional fresh vertices having collective weight $4$; see Figure~\ref{fig:make-a-single-grid-graph-cubic-by-just-adding-three-vertices}).
    In total, eleven fresh edges are added during this construction.
    As only new vertices labeled $1$, both already defended by new vertices labeled $2$, are neighbored to the initial graph $G_{2,m}$, the SDRDF requirements are satisfied.
    By cubicity and using the bound \eqref{eq:lowest-bound-scope-cubic-graphs}, this implies that $\signedDoubleRomanDominationNumber{G_{2,m}} + 4 = \cumulativeWeightFuncargfree{V} + 4 = \cumulativeWeightFuncargfree{\tilde{V}} \geqslant |\tilde{V}|/2+1 = (2m+6)/2+1$, and consequently $\signedDoubleRomanDominationNumber{G_{2,m}}\geqslant m$.

    The rest of the proof is now dedicated to the case $m\not\equiv 0 \pmod{2}$.
    Let $R:=\{u_{m-2},\allowbreak{}u_{m-1},\allowbreak{}v_{m-2},\allowbreak{}v_{m-1}\}$.
    In the following, we consider the sequence of vertices $s_p:=\langle w_0,\allowbreak{}\ldots,\allowbreak{} w_{p-1};\allowbreak{}x_0,\allowbreak{}\ldots,\allowbreak{}\allowbreak{}x_{p-1}\rangle$, $p = 4,6$, to which we want to associate a respective sequence of labels.
    These vertices will be part of a $2\times p$ grid graph $H_p$, which will be connected to our studied grid graph $G_{2,m}$, see Figure~\ref{fig:add-eight-nodes-to-make-prims-graph-for-bound-proof}.
    The argumentation for the lower bound $m$ respectively $m+1$ of $\signedDoubleRomanDominationNumber{G_{2,m}}$ is split into several cases, depending on the distribution of the vertices labeled $-1$ inside $R$, in which we extend $G_{2,m}$ to a suitably labeled version of $\generalizedPetersen{m+p}{1}$ when needed.
    In each of the following cases, the claimed bound holds.
    Note that it is enough, by condition \eqref{eq:sdrdp-positive-closed-neighborhood-condition}, to consider at most two vertices in $V_{-1}\cap R$.
    
    \emph{Case 1.} $|V_{-1}\cap R| = 1$.
    
    \emph{Subcase 1.1.} $V_{-1} \cap R \in \{\{u_{m-1}\},\{v_{m-1}\}\}$.
    W.l.o.g. $V_{-1} \cap R = \{v_{m-1}\}$.
    \begin{itemize}
      \item If $v_{m-1}$ is defended by its lower $3$-labeled neighbor $u_{m-1}$, then in the extended graph in Figure~\ref{fig:add-eight-nodes-to-make-prims-graph-for-bound-proof}, for $p=4$, we choose for $s_4$ the sequence of labels $\langle -1,\allowbreak{}-1,\allowbreak{}2,\allowbreak{}1;\allowbreak{}1,\allowbreak{}2,-1,\allowbreak{}2\rangle$, yielding additional weight $5$.
      By this we get a SDRDF on $P_{m+4,1}$ with weight of $\signedDoubleRomanDominationNumber{G_{2,m}}+5$, which implies the inequality $\signedDoubleRomanDominationNumber{G_{2,m}}+5\geqslant \signedDoubleRomanDominationNumber{\generalizedPetersen{m+4}{1}}$.
      Since, for $m\equiv 1 \pmod 4$ by Theorem \ref{thm:generalized-petersen-graph-case-k-equal-one} we have $\signedDoubleRomanDominationNumber{\generalizedPetersen{m+4}{1}} = (m+4)+2=m+6$, we obtain $\signedDoubleRomanDominationNumber{G_{2,m}}+5 \geqslant m+6$, i.e., $\signedDoubleRomanDominationNumber{G_{2,m}}\geqslant m +1$.
      
      On the other hand, for $m\equiv 3 \pmod 4$ also by Theorem~\ref{thm:generalized-petersen-graph-case-k-equal-one} we have $\signedDoubleRomanDominationNumber{\generalizedPetersen{m+4}{1}}=(m+4)+1=m+5$, which yields $\signedDoubleRomanDominationNumber{G_{2,m}}+5\geqslant m+5$, i.e., $\signedDoubleRomanDominationNumber{G_{2,m}}\geqslant m$.
      
      \item If $v_{m-1}$ is defended by its left $3$-labeled neighbor $v_{m-2}$, the labeling of $G_{2,m}$ even cannot be optimal:
      either $u_{m-1}$ is an unnecessary defender, or $u_{m-1}$ is labeled $1$ which implies that $u_{m-2}$ has a label from $\{2,3\}$ in turn implying that $\langle u_{m-2}, u_{m-1}\rangle$ should have received labels $\langle 3, -1\rangle$ to reduce weight.
      \item Also the scenario of purely $2$-labeled neighbors of $v_{m-1}$ has to be considered: Recall that the label of $u_{m-2}$ is positive by assumption.
      Hence, we can relabel $R$ such that $\{u_{m-2},v_{m-2}\}\subseteq V_3$ and $\{u_{m-1},v_{m-1}\}\subseteq V_{-1}$.
      By Observation~\ref{obs:small-inductive-proof-remaining-case} in the appendix, we know that the latter boundary constraints imply that $\cumulativeWeight{G_{2,m}}{f}$ cannot under-run the bound $m+1$ respectively $m$ when $m\equiv 1\pmod{4}$ respectively $m\equiv 3\pmod{4}$.
      Therefore, we do not need to come up with another construction here.
    \end{itemize}
    
    \emph{Subcase 1.2.} $V_{-1} \cap R \in \{\{u_{m-2}\},\{v_{m-2}\}\}$.
    W.l.o.g. let $V_{-1} \cap R = \{v_{m-2}\}$.
    We can just add the connecting edges $u_{m-1}u_0$ and $v_{m-1}v_0$.
    By positivity of the righter-most labels in $R$, this fulfills all SDRDF constraints at no additional weight cost.

    \emph{Case 2.} $|V_{-1}\cap R| = 0$. Replicate the construction of Subcase 1.2.

    \emph{Case 3.} $|V_{-1}\cap R| = 2$.
    
    \emph{Subcase 3.1.} Horizontal occurrences, i.e., $V_{-1}\cap R \in\{\{u_{m-2}, u_{m-1}\},\{v_{m-2},v_{m-1}\}\}$.
    W.l.o.g. assume $V_{-1}\cap R = \{v_{m-2}, v_{m-1}\}$.
    In this subcase, we proceed as in the first paragraph of Subcase 1.1
    (in the extended graph in Figure~\ref{fig:add-eight-nodes-to-make-prims-graph-for-bound-proof} the sequence $s_4$ shall have associated labels $\langle -1,\allowbreak{}-1,\allowbreak{}2,\allowbreak{}1;\allowbreak{}1,\allowbreak{}2,-1,\allowbreak{}2\rangle$).
    Note that $\langle u_{m-2}, u_{m-1}\rangle$ must necessarily have the labels $\langle x, 3\rangle$ where $x\geqslant 1$.
    Clearly, despite $w_3,x_3$ in $H_4$ are joined potentially both with vertices labeled $-1$, they will not violate condition \eqref{eq:sdrdp-positive-closed-neighborhood-condition} as abundantly defended.
    Therefore, the entire labeling is a SDRDF having an additional weight cost of $5$ due to the vertices in $H_4$.
    
    \emph{Subcase 3.2.} Vertical occurrences (interior), i.e., $V_{-1}\cap R = \{u_{m-2}, v_{m-2}\}$.
    We note that at least one label of the necessarily positively labeled vertices $u_{m-1}$, $v_{m-1}$ must further have assigned label $2$ or $3$ -- w.l.o.g. assume $v_{m-1}\in V_{2}\cup V_{3}$ and $u_{m-1}\in V_{1}\cup V_2\cup V_3$.

    We now consider the extended graph in Figure~\ref{fig:add-eight-nodes-to-make-prims-graph-for-bound-proof}, where for the sequence of vertices $s_4$, we pick the sequence of labels $\langle -1,\allowbreak{}3,-1,\allowbreak{}1;\allowbreak{} -1,\allowbreak{}2,\allowbreak{}-1,\allowbreak{}2\rangle$, costing additional weight $4$.
    We now prove this subcase using Theorem~\ref{thm:generalized-petersen-graph-case-k-equal-one} as in Subcase 1.1.
    
    \emph{Subcase 3.3.} Vertical occurrences (righter-most), i.e., $V_{-1}\cap R = \{u_{m-1}, v_{m-1}\}$.
    Necessarily, we have that $u_{m-2},v_{m-2}\in V_3$.
    This situation is observed in the third paragraph of Subcase 1.1 and concluded by Observation~\ref{obs:small-inductive-proof-remaining-case}.
    
    \emph{Subcase 3.4.} Diagonal occurrences, i.e., $V_{-1}\cap R \in\{\{u_{m-2}, v_{m-1}\},\{v_{m-2},u_{m-1}\}\}$.
    W.l.o.g. assume $V_{-1}\cap R = \{u_{m-2}, v_{m-1}\}$.
    Note that $\langle v_{m-2}, u_{m-1} \rangle $ must necessarily have associated label sequence $\langle x,3\rangle$ where $x\geqslant 1$.
    For $x\geqslant 2$, in the extended graph in Figure~\ref{fig:add-eight-nodes-to-make-prims-graph-for-bound-proof}, for $p=4$, pick for $s_4$ the sequence of labels $\langle 1,\allowbreak{}-1,-1,\allowbreak{}3;\allowbreak{} -1,\allowbreak{}3,\allowbreak{}1,\allowbreak{}1\rangle$, costing additional weight $6$.
    Finally, we update the label value of $u_{m-1}$ to $2$ (not violating the SDRDF constraints).
    Hence, finally, we obtain a graph $\generalizedPetersen{m+4}{1}$ costing additional weight $5$ and conclude this subcase again as in the first part of Subcase 1.1.

    For the case $x=1$ we observe how vertices $v_0,v_1,u_0,u_1$ are labeled.
    \begin{itemize}
      \item If neither $\{v_1,u_0 \}\subseteq V_{-1}$ nor $\{v_0,u_1\}\subseteq V_{-1}$, i.e., we do not have a diagonal of vertices in $V_{-1}$ on the left side of $G_{2,m}$, then for the horizontally flipped labeling\footnote{Formally we substitute each label of $u_i$ and $v_i$ by the label of $u_{m-1-i}$ and $v_{m-1-i}$, respectively, $i=0,\ldots,m-1$.
        Bounds proven for this labeling clearly also apply for the non-flipped variant of the labeling.} the claim follows directly from one of the previously settled (sub)cases 1, 2, 3.1, 3.2, or 3.3 of this proof.
      
      \item If $\{u_0,v_1\}\subseteq V_{-1}$, then necessarily $v_0\in V_3$ and $u_1\not\in V_{-1}$.
      Hence, making use of the construction given in Figure~\ref{fig:add-eight-nodes-to-make-prims-graph-for-bound-proof} ($p=6$) to extend the graph $G_{2,m}$ to $\generalizedPetersen{m+6}{1}$, where we associate the sequence of labels $\langle \allowbreak{}1, \allowbreak{}-1, -1,\allowbreak{}3,\allowbreak{}3,\allowbreak{}-1;\allowbreak{}-1, \allowbreak{}3, \allowbreak{}1,\allowbreak{}-1,\allowbreak{}-1,\allowbreak{}1\rangle$ to $s_6$,
      we obtain a SDRDF on $\generalizedPetersen{m+6}{1}$ of total weight $\signedDoubleRomanDominationNumber{G_{2,m}}+6$.
      For $m\equiv 1 \pmod 4$ by Theorem \ref{thm:generalized-petersen-graph-case-k-equal-one} we have $\signedDoubleRomanDominationNumber{\generalizedPetersen{m+6}{1}} = (m+6)+1=m+7$, which implies $\signedDoubleRomanDominationNumber{G_{2,m}}+6\geqslant m+7$, i.e. $\signedDoubleRomanDominationNumber{G_{2,m}}\geqslant m +1$.
      On the other hand, for $m\equiv 3 \pmod 4$ also by Theorem~\ref{thm:generalized-petersen-graph-case-k-equal-one} we have $\signedDoubleRomanDominationNumber{\generalizedPetersen{m+6}{1}}=(m+6)+2=m+8$, which yields $\signedDoubleRomanDominationNumber{G_{2,m}}+6\geqslant m+8$, i.e., $\signedDoubleRomanDominationNumber{G_{2,m}}\geqslant m +2>m$.
      \item If $\{u_1,v_0\}\subseteq V_{-1}$, then necessarily $u_0,v_1\not\in V_{-1}$.
      Hence we can add the edges $u_{m-1}u_0$, $v_{m-1}v_0$ to $G_{2,m}$ obtaining a SDRDF on $\generalizedPetersen{m}{1}$.
    \end{itemize}
  \end{proof}
  
  \begin{pro}\label{pro:flower-snarks-ub}
    For $\flowerSnarks{m}$, $m\geqslant 5$, we have $2m\leqslant\signedDoubleRomanDominationNumber{\flowerSnarks{m}} \leqslant 2m + 1$.
  \end{pro}
  \begin{proof}    
    Let us first show the validity of the \emph{upper bound}, i.e. $\signedDoubleRomanDominationNumber{\flowerSnarks{m}} \leqslant 2m + 1, m\geqslant 5$.
    
    \emph{Case 1.} $m\equiv 0\pmod{3}$.
    
    We choose the labeling with $V_1=\{a_{m-1},c_{m-1}\}$, $V_2 = \{b_{3i}, b_{3i+1},c_{3i+1},d_{3i},d_{3i+2}\mid i = 0,1,\ldots,\frac{m-3}{3}\}$ $\cup\{c_{3i+2}\mid i = 0,1,\ldots,\frac{m-6}{3}\}$, and $V_{-1} = V\setminus(V_1\cup V_2)=$ $\{b_{3i+2},c_{3i},d_{3i+1}\mid i = 0,1,\ldots,\frac{m-3}{3}\}\cup\{a_i\mid i = 0,1,\ldots m-2\} $; for $m = 9$, this is illustrated in Figure \ref{fig:fs9}.
    One can easily check that the SDRDF properties are satisfied.
    Consequently, we have $\signedDoubleRomanDominationNumber{\flowerSnarks{m}} \leqslant \cumulativeWeight{\flowerSnarks{m}}{f} = 2|V_2|+|V_1|-|V_{-1}| = 2(2m-1)+2-(2m-1) = 2m+1$.
    
    \begin{figure}[t]
      \centering
      \includegraphics{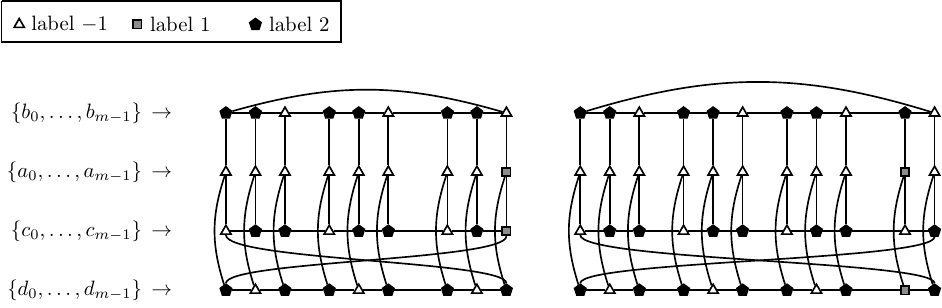}
      \caption{SDRDFs for $\flowerSnarks{m}$ when $m=9$ (left) respectively $m=11$ (right).
        Thinking of the vertices as placed on a grid, a labeling pattern of dimensions $4\times 3$, periodically repeated and finally flanked by an individual termination pattern of dimensions $4\times 3$ (left), respectively $4\times 2$ (right), can be read off.
        These labeling patterns generalize to higher values of $m$ of congruency $m\equiv0\pmod{3}$ and $m\equiv 2\pmod{3}$, respectively.
      }
      \label{fig:fs9}
    \end{figure}
    
    \emph{Case 2.} $m\equiv 1 \pmod{3}$.
    
    We pick the labeling with $V_1=\{a_{m-1},b_{m-1}\}$, $V_2 = \{b_{3i}, b_{3i+2},c_{3i},c_{3i+1},d_{3i+1},d_{3i+2}\mid i = 0,1,\ldots,\frac{m-4}{3}\}$ $\cup\{c_{m-1}\}$, and $V_{-1} = V\setminus(V_1\cup V_2)=$ $\{a_i\mid i = 0,1,\ldots,m-2\}\cup\{b_{3i+1},c_{3i+2},d_{3i}\mid i = 0,1,\ldots \frac{m-4}{3}\}\cup\{d_{m-1}\}$; for $m = 13$, this is illustrated in Figure~\ref{fig:fs13}.
    Again one can quickly check that $f$ is indeed a SDRDF.
    Therefore, $\cumulativeWeight{\flowerSnarks{m}}{f} = 2|V_2|+|V_1|-|V_{-1}| = 2(2m-1)+2-(2m-1) = 2m+1$ is an upper bound for $\signedDoubleRomanDominationNumber{\flowerSnarks{m}}$.

    \begin{figure}[h]
      \centering
      \includegraphics{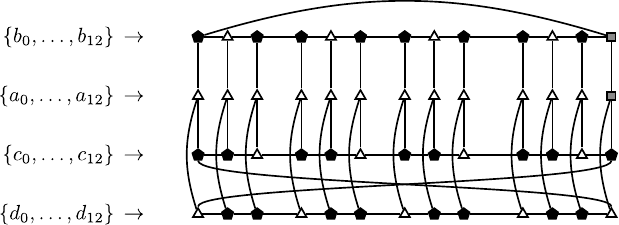}
      \caption{A SDRDF for the graph $\flowerSnarks{13}$ (information displayed as in Figure~\ref{fig:fs9}).
        Again the labeling scheme, consisting of a periodically repeating $4\times 3$ pattern of labels, which is flanked by a terminating $4\times 4$ pattern of labels, naturally generalizes to higher values of $m\equiv 1\pmod{3}$.}
      \label{fig:fs13}
    \end{figure}
    \emph{Case 3.} $m\equiv 2\pmod{3}$.
    
    Choose the labeling with $V_1=\{a_{m-2},d_{m-2}\}$, 
    $V_2 = \{\allowbreak{}b_{3i}, \allowbreak{}b_{3i+1},\allowbreak{}c_{3i+1},\allowbreak{}c_{3i+2},\allowbreak{}d_{3i},\allowbreak{}d_{3i+2}\mid i = 0,1,\ldots,\frac{m-5}{3}\}$ $\allowbreak{}\cup\{\allowbreak{}b_{m-2},\allowbreak{}c_{m-1},\allowbreak{}d_{m-1}\}$, and $V_{-1} = V\setminus(V_1\cup V_2)=$ $\{a_i\mid i = 0,1,\ldots,{m-3},m-1\}\cup\{\allowbreak{}b_{3i+2},\allowbreak{}c_{3i},\allowbreak{}d_{3i+1}\mid i = 0,1,\ldots \frac{m-5}{3}\}\cup\{a_{m-1},b_{m-1},c_{m-2}\}$; for $m = 11$,  this is illustrated in Figure~\ref{fig:fs9}.
    One can easily see that $f$ is indeed a SDRDF, implying
    $\signedDoubleRomanDominationNumber{\flowerSnarks{m}}\leqslant \cumulativeWeight{\flowerSnarks{m}}{f} = 2|V_2|+|V_1|-|V_{-1}| = 2(2m-1)+2-(2m-1) = 2m+1$.
    
    \medskip
    
    Concerning the \emph{lower bound}, we obtain $\signedDoubleRomanDominationNumber{\flowerSnarks{m}} \geqslant 2m$ from Theorem~\ref{thm:cubic-graph-perfectly-sharp-bound}, which concludes the proof.
  \end{proof}
  
  \section{Conclusions and future work}\label{sec:conclusions}
  
  In this work, we studied the signed Roman domination problem on cubic graphs in detail.
  The discharging method turned out to be a powerful tool allowing us to come up with a sharp lower bound.
  In this context, we were able to take advantage of some findings on $\alpha$-total domination and thus improve the upper bound.
  Moreover, we emphasized the importance of generalized Petersen graphs as paramount examples of cubic graphs attaining this best possible lower bound.
  We have presented a constraint programming driven approach that seems adaptable to several other classes of rotationally symmetric graphs, and furthermore can easily be applied to other forms of domination.
  
  The achieved results form the foundation for several interesting future research questions.
  In addition to the obtained sharp lower bound for $\signedDoubleRomanDominationNumberArgfree$ on cubic graphs, it would be interesting to find a sharp upper bound.
  Proving a sharp asymptotic upper bound might be interesting, too.
  We here mean to study, given a class of graphs $\mathcal{G}$ of unbounded order, the quantity
  \begin{equation}
    \sizerenormalizedSDRDNumber{\mathcal{G}} := \limsup_{\substack{G\in\mathcal{G},G=(V,E) \\ |V|\to\infty }} {|V|}^{-1}\signedDoubleRomanDominationNumber{G}.\label{eq:asymptotic-upperbound-version}
  \end{equation}
  Slightly differing from a related quantity studied by Egunjobi and Haynes \cite[p.~72]{egunjobi2020perfect}, the latter captures the behavior of the maximum per-vertex average weight when graph sizes are supposed to grow, therefore neglecting all small graphs of high average weight.
  
  By Proposition~\ref{pro:general-upper-bound-on-cubic-graphs-by-alpha-domination-proof}, we already know that $\sizerenormalizedSDRDNumber{\mathcal{C}}\leqslant 5/4$ for the class $\mathcal{C}$ of cubic graphs; this bound is, however, unlikely to be sharp.
  Identifying subclasses $\mathcal{C}'$ of cubic graphs having maximum $\sizerenormalizedSDRDNumber{\mathcal{C}'}$-value seems challenging.
  In this regard, we make the following observation.
  \begin{obs}
    There are subclasses $\mathcal{C}'$ of cubic graphs, for which $\sizerenormalizedSDRDNumber{\mathcal{C}'}\geqslant 7/10$.
    In particular, $\sizerenormalizedSDRDNumber{\mathcal{C}}\geqslant 7/10$.
  \end{obs}
  \begin{proof}
    Let $\mathcal{C}'$ contain all graphs $G_k$, $k\in\mathbb{N}\setminus\{0\}$, where $G_k$ is made up by $k$ connected components all being isomorphic to $\generalizedPetersen{5}{1}$ (cubic).
    Each graph $G_k$ consists of $n=10k$ vertices and has SDRDF weight $7k$.
    Consequently, $\sizerenormalizedSDRDNumber{\mathcal{C}'}=7/10$.
  \end{proof}
  
  If we set our attention on the class $\mathcal{C}_{\textrm{conn}}$ of \emph{connected} cubic graphs, the dynamic might change, and we pose ourselves the following question.
  
  \begin{problem}
    \begin{enumerate}[label={(\roman*)}]
      \item How large can $\rho > 1/2$ be chosen such that $\sizerenormalizedSDRDNumber{\mathcal{C}_{\textrm{conn}}}\geqslant \rho$?
      \item Is it possible that $\sizerenormalizedSDRDNumber{\mathcal{C}_{\textrm{conn}}} \geqslant 9/16$?\label{ite:extrmal-property-conjectured-for-connected-cubic-graphs}
      \item Do the graphs $\generalizedPetersen{m}{2}$ attain the bound in \ref{ite:extrmal-property-conjectured-for-connected-cubic-graphs} (such an average weight is attained for $m=8,16$)?
    \end{enumerate}
  \end{problem}
  
  In preliminary work, we constructed optimal SDRDFs for $2\times m$ grid graphs, and for paths of length $m$ such graphs have been determined in~\cite{abdollahzadehahangar2019signeddoubledominationINgraphs}.
  This naturally raises the following challenge concerning general $\ell \times m$ grid graphs.
  
  \begin{problem}\label{probl:general-grid-graphs-number}
    Determine $\signedDoubleRomanDominationNumberArgfree$ on $\ell \times m$ grid graphs for further (small) values $\ell\in\mathbb{N}$ and general $m\in\mathbb{N}$.
  \end{problem}

  For solving Problem~\ref{probl:general-grid-graphs-number} it might be a reasonable strategy to obtain sharp bounds for $\signedDoubleRomanDominationNumberArgfree$ on $4$-regular graphs.
  Moreover, the fact that the signed domination problem is NP-hard on grids \cite{zheng2013kernelization} leads to the following question when $\ell$ is kept general.
  
  \begin{problem}
    Is it NP-hard to determine the existence of a SDRDF on an $\ell \times m$ grid graph with a weight not exceeding a given limit?
  \end{problem}
  
  From our experience in the setting of the SDRDP, the requirement of a particular ``balance'' of defenders and defendants, as well as the higher flexibility on how to defend, make it challenging in comparison to the domination-type problems mentioned earlier.
  
  \section*{Acknowledgments}
  Enrico Iurlano and Günther Raidl are supported by Austria's Agency for Education and Internationalization under grant BA05/2023.
  Tatjana Zec and Marko Djukanovi\'c are supported by the bilateral project between Austria and Bosnia and Herzegovina funded by
  the Ministry of Civil Affairs of Bosnia and Herzegovina under grant no.\ 1259074.
  Moreover, this project is partially funded by the Doctoral Program ``Vienna Graduate School on Computational Optimization'', Austrian Science Fund (FWF), grant W1260-N35.

  \bibliographystyle{elsarticle-num}

  \clearpage
  \newpage
  
  \appendix\label{appendix:A}

  \section{Lookup tables: satisfaction of SDRDP constraints for $\generalizedPetersen{m}{3}$}\label{sec:verification-of-gp-m-three-and-co-labelings}
  
  In Table~\ref{tab:Pm3_case1} respectively Table~\ref{tab:Pm3_case2}, the fact that the function $f$ defined in Theorem~\ref{thm:generalized-petersen-graph-case-k-equal-three}, Case 1 respectively Case 2 is a SDRDF can be read off.
  \begin{table}[H]
    \centering
    \begin{minipage}[t]{.5\linewidth}
      \vspace{-\topskip}
      \centering
      \begin{longtable*}[l]{p{3.88cm}|l}
        $v\in V_{-1}$ & Defenders of $v$ \\\hline\hline
        $u_{4i+2}, i\in \{0,1,\ldots,\frac{m-9}4\}$ & $u_{4i+1},v_{4i+2}$ \\\hline
        $u_{4i+3}, i\in \{0,1,\ldots,\frac{m-9}4\}$ & $u_{4i+4},v_{4i+3}$ \\\hline
        $u_{m-3}$ & $u_{m-4},u_{m-2}$ \\\hline
        $u_{m-1}$ & $u_{0},u_{m-2}$ \\\hline
        $v_{4i}, i\in \{0,1,\ldots,\frac{m-9}4\}$ & $u_{4i},v_{4i+3}$ \\\hline
        $v_{4i}, i\in \{0,1,\ldots,\frac{m-5}4\}$ & $u_{4i},v_{4i+3}$ \\\hline
        $v_{4i+1}, i\in \{0,1,\ldots,\frac{m-5}4\}$ & $u_{4i+1},v_{4i-2}$
      \end{longtable*}
      
      \begin{longtable*}[l]{p{3.88cm}|l}
        $v\in V_1$ &  Defenders of $v$ \\\hline\hline
        $v_{m-3}$ & $v_{m-6}$ \\\hline
        $v_{m-1}$ & $v_{2}$ \\\hline
      \end{longtable*}
    \end{minipage}
    \begin{minipage}[t]{.5\linewidth}
      \vspace{-\topskip}
      \begin{longtable*}[r]{l|l}
        $v\in V_{2}$ & \parbox{2.7cm}{Two vertices in\\
          $N[v]\cap (V\setminus V_{-1})$}\\\hline\hline
        $u_{4i}, i\in \{0,1,\ldots,\frac{m-5}4\}$ & $u_{4i},u_{4i+1}$ \\\hline
        $u_{4i+1}, i\in \{0,1,\ldots,\frac{m-5}4\}$ & $u_{4i},u_{4i+1}$ \\\hline
        $u_{m-2}, i\in \{0,1,\ldots,\frac{m-5}4\}$ & $v_{m-2}$ \\\hline
        $v_{4i+2}, i\in \{0,1,\ldots,\frac{m-9}4\}$ & $v_{4i-1},v_{4i+2}$ \\\hline
        $v_{4i+3}, i\in \{0,1,\ldots,\frac{m-9}4\}$ & $v_{4i+3},v_{4i+6}$ \\\hline
        $v_{m-2}, i\in \{0,1,\ldots,\frac{m-5}4\}$ & $u_{m-2}$ \\\hline
      \end{longtable*}
    \end{minipage}
    
    \caption{Fulfillment of \eqref{eq:sdrdp-existing-defender-for-minusone-condition}--\eqref{eq:sdrdp-existing-defender-for-one-condition} respectively \eqref{eq:replacement-condition} for $f$ defined in Theorem~\ref{thm:generalized-petersen-graph-case-k-equal-three}, Case 1.
      Here the validity of condition \eqref{eq:replacement-condition} is often given implicitly: 
      Note that since \eqref{eq:sdrdp-existing-defender-for-one-condition} holds, the property \eqref{eq:replacement-condition} automatically applies for vertices in $V_1$.
      Since the condition \eqref{eq:sdrdp-existing-defender-for-minusone-condition} holds and $V_3=\emptyset$, we have that the property \eqref{eq:replacement-condition} holds for all vertices in $V_{-1}$.
      Therefore, for all vertices in $V_2$ it remains to check validity of \eqref{eq:replacement-condition}, which can be read off the right table.}\label{tab:Pm3_case1}
  \end{table}
  
  \begin{table}[H]
    \centering
    \begin{minipage}[t]{.5\linewidth}
      \vspace{-\topskip}
      \centering
      \begin{longtable*}[l]{p{4cm}|p{2.4cm}}
        $v\in V_{-1}$ & Defenders of $v$ \\\hline\hline
        $u_{4i}, i\in \{0,1,\ldots,\frac{m-11}4\}$ & $u_{4i-1},v_{4i}$ \\\hline
        $u_{4i+1}, i\in \{0,1,\ldots,\frac{m-11}4\}$ & $u_{4i+2},v_{4i+1}$ \\\hline
        $u_{m-8}$ & $u_{m-9},u_{m-7}$ \\\hline
        $u_{m-6}$ & $u_{m-7},u_{m-5}$ \\\hline
        $u_{m-4}$ & $u_{m-5},v_{m-4}$ \\\hline
        $u_{m-3}$ & $v_{m-3}$ \\\hline
        $v_{4i+2}, i\in \{0,1,\ldots,\frac{m-15}4\}$ & $u_{4i+2},v_{4i+5}$ \\\hline
        $v_{4i+3}, i\in \{0,1,\ldots,\frac{m-15}4\}$ & $u_{4i+3},v_{4i}$ \\\hline
        $v_{m-8}$ & $v_{m-11},v_{m-5}$ \\\hline
        $v_{m-6}$ & $v_{m-3}$ \\\hline
        $v_{m-2}$ & $v_{1},v_{m-5}$ \\\hline
        $v_{m-1}$ & $u_{m-1},v_{m-4}$ \\\hline
      \end{longtable*}
      \begin{longtable*}[l]{p{4cm}|p{2.4cm}}
        $v\in V_1$  & Defenders of $v$ \\\hline\hline
        $u_{m-2}$ & $u_{m-1}$ \\\hline
        $v_{m-9}$ & $u_{m-9}$ \\\hline
        $v_{m-7}$ & $u_{m-7}$ ($v_{m-10},v_{m-4}$) \\\hline
      \end{longtable*}%
    \end{minipage}%
    \begin{minipage}[t]{.5\linewidth}
      \vspace{-\topskip}
      \begin{longtable*}[r]{l|p{2.3cm}}
        $v\in V_{2}$ & \parbox{2.3cm}{Two vertices in\\
          $N[v]\cap (V\setminus V_{-1})$}\\\hline\hline
        $u_{4i+2}, i\in \{0,1,\ldots,\frac{m-15}4\}$ & $u_{4i+2},u_{4i+3}$ \\\hline
        $u_{4i+3}, i\in \{0,1,\ldots,\frac{m-15}4\}$ & $u_{4i+2},u_{4i+3}$ \\\hline
        $u_{m-9}$ & $u_{m-9},v_{m-9}$ \\\hline
        $u_{m-7}$ & $u_{m-7},v_{m-7}$ \\\hline
        $u_{m-5}$ & $u_{m-5},v_{m-5}$ \\\hline
        $u_{m-1}$ & $u_{m-2},u_{m-1}$ \\\hline
        $v_{4i}, i\in \{0,1,\ldots,\frac{m-11}4\}$ & $v_{4i-3},v_{4i}$ \\\hline
        $v_{4i+1}, i\in \{0,1,\ldots,\frac{m-11}4\}$ & $v_{4i+1},v_{4i+4}$ \\\hline
        $v_{m-5}$ & $u_{m-5},v_{m-5}$ \\\hline
        $v_{m-4}$ & $v_{m-7},v_{m-4}$ \\\hline
      \end{longtable*}
      
      \begin{longtable*}[r]{l|p{2.3cm}}
        $v\in V_{3}$ &
        \parbox{2.7cm}{Two vertices in\\
          $N[v]\cap (V\setminus V_{-1})$}\\\hline\hline
        $v_{m-3}$ & $v_{0},v_{m-3}$ \\\hline
      \end{longtable*}
    \end{minipage}
    \caption{Fulfillment of \eqref{eq:sdrdp-existing-defender-for-minusone-condition}--\eqref{eq:sdrdp-existing-defender-for-one-condition} respectively \eqref{eq:replacement-condition} for $f$ defined in Theorem~\ref{thm:generalized-petersen-graph-case-k-equal-three}, Case 2.
      Note that, as in Case 1, the vertices in $V_1$ satisfy \eqref{eq:replacement-condition}.
      The vertices in $V_{-1}$ which are defended by $v_{m-3}$ are $u_{m-3}$ and $v_{m-6}$, and they are also adjacent to $u_{m-2}\in V_1$ respectively $v_{m-9}\in V_1$.
      The remaining vertices in $V_{-1}$ are defended by two vertices in $V_2$.
      Hence, all vertices in $V_{-1}$ fulfill \eqref{eq:replacement-condition}.
      The tables on the right testify that  \eqref{eq:replacement-condition} is valid for the vertices in $V_2\cup V_3$.}
    \label{tab:Pm3_case2}
  \end{table}

  \section{Results by constraint programming}
  
  Algorithm~\ref{alg:exhaustively-comparing-optima} shows how the results in Lemma~\ref{lem:constraing-programming-as-mathematical-lemma-formulation} were obtained.
  We generated the models for the constraint programming framework \texttt{MiniZinc}\footnote{\url{https://www.minizinc.org/}} in version 2.7.4, which in turn was configured to use the solver \texttt{Chuffed}\footnote{\url{https://github.com/chuffed/chuffed}} in version 0.12.0 to determine the minima for the encountered optimization problems.
  Jointly with several observations on it used in the present paper, the database returned by Algorithm~\ref{alg:exhaustively-comparing-optima} is available online\footnote{\url{https://www.ac.tuwien.ac.at/files/resources/instances/sdrdp/queries_sdrdp.zip}}.
  
  \begin{breakablealgorithm}
    \caption{Exhaustively comparing optima}\label{alg:exhaustively-comparing-optima}
    \begin{algorithmic}[1]
      \Statex \textbf{Input}: {Constraint programming solver $S$; empty database $\mathit{DB}$}
      \Statex \textbf{Output}: {Populated database $\mathit{DB}$}
      \Statex
      \State $C \gets \{u_{i},v_{i}\mid i=0,\ldots,7\}$, $C' \gets C\setminus\{u_{i},v_{i}\mid i=4,\ldots,7\}$
      \State $L \gets \{\nodeLT,\nodeLTI,\nodeLB,\nodeLBI\}$, $R \gets \{\nodeRTI,\nodeRT,\nodeRBI,\nodeRB\}$
      \State Let $G$ be the graph of Figure~\ref{fig:two-by-twelve-grid-graph} with set of vertices $L\cup C \cup R$
      \State Let $G'$ be the graph of  Figure~\ref{fig:two-by-eight-grid-graph-after-cutoff} with set of vertices $L\cup C' \cup R$
      \Statex
      \State $Q\gets [\,]$ \texttt{ //already handled constellations modulo symmetry breaking}
      \Statex
      \For{\textbf{each} $d = (d_0,\ldots,d_7) \text{ in } \labeldomainSet^8$}
      
      \IIf{$d$ contained in $Q$ modulo symmetry breaking} \textbf{continue}\EndIIf
      \IIf{$d$ places more than two labels $-1$ on $L$ or on $R$} \textbf{continue} \texttt{/*infeasible*/} \EndIIf
      \State $Q.\mathrm{add}(d)$
      \Statex
      \State Clear all label constraints $\nodeLT,\nodeLTI, \nodeLB,\nodeLBI, \nodeRTI,\nodeRT,\nodeRBI,\nodeRB$ for $S$
      \State S.add\_constraint($\langle \nodeLT,\nodeLTI, \nodeLB,\nodeLBI, \nodeRTI,\nodeRT,\nodeRBI,\nodeRB \rangle = \langle d_0,\ldots,d_7\rangle$)
      \For{\textbf{each} $u\text{ in } L\cup C \cup R$}
      \State $\mathrm{S}.\mathrm{add\_constraint}(\text{label of }u\text{ is } \labeldomainSet\text{-valued})$
      \If{$u\text{ not in }\{\nodeLT,\nodeLB,\nodeRT,\nodeRB\}$} \texttt{//ignores corners}
      \State $\mathrm{S}.\mathrm{add\_constraint}(u\text{ satisfies \eqref{eq:sdrdp-existing-defender-for-minusone-condition}--\eqref{eq:sdrdp-positive-closed-neighborhood-condition} w.r.t. adjacency of }G)$
      \EndIf
      \EndFor
      \State $\mathrm{minweight\_C} \gets \mathrm{S}.\mathrm{minimize}()$
      \Statex
      \State Clear all label constraints of  $\nodeLT,\nodeLTI, \nodeLB,\nodeLBI, \nodeRTI,\nodeRT,\nodeRBI,\nodeRB$ for $S$
      \State $\mathrm{S}.\mathrm{add\_constraint}(\langle \nodeLT,\nodeLTI, \nodeLB,\nodeLBI, \nodeRTI,\nodeRT,\nodeRBI,\nodeRB \rangle = \langle d_0,\ldots,d_7\rangle)$
      \For{\textbf{each}  $u\text{ in } L\cup C' \cup R$}
      \State $\mathrm{S}.\mathrm{add\_constraint}(\text{label of }u\text{ is } \labeldomainSet\text{-valued})$
      \If{$u\text{ not in }\{\nodeLT,\nodeLB,\nodeRT,\nodeRB\}$}
      \State $\mathrm{S}.\mathrm{add\_constraint}(u\text{ satisfies \eqref{eq:sdrdp-existing-defender-for-minusone-condition}--\eqref{eq:sdrdp-positive-closed-neighborhood-condition} w.r.t. adjacency of }G')$
      \EndIf
      \EndFor
      \State $\mathrm{minweight\_Cprime}\gets \mathrm{S}.\mathrm{minimize}()$
      \Statex
      \State $\mathrm{delta}\gets\mathrm{minweight\_C} - \mathrm{minweight\_Cprime}$
      \State $\mathit{DB}.\mathrm{insert}(\langle d_0,\ldots,d_7\rangle, \mathrm{minweight\_C}, \mathrm{minweight\_Cprime}, \mathrm{delta})$
      \IIf{$\mathrm{delta}\geqslant 4$}
      \textit{print}(``Quality-transferring constellation found:'', $d$)
      \EndIIf
      \EndFor
    \end{algorithmic}
  \end{breakablealgorithm}
  
  \medskip
  
  \begin{obs}\label{obs:thoughts-with-sufficently-many-flanking-columns}
    For $m+4=8\ell +r$ with $m+4\geqslant 17$ and $r\in\{1,5\}$, consider $\generalizedPetersen{m+4}{1}$ with an optimal SDRDF $f$ defined on it.
    Let $W_{\textrm{bdry}} := \cumulativeWeight{\{u_{-1},\allowbreak{}u_{-2},\allowbreak{} v_{-1}, v_{-2}\}\allowbreak{}\cup\allowbreak{}\{u_8, v_8,\allowbreak{}u_9,\allowbreak{}v_9\}}{f}$ and assume $W_{\textrm{cntr}} := \cumulativeWeight{\{u_j,v_j\mid j=0,\ldots,7\}}{f}\geqslant 10$.
    Set $W_{t} := W_{\textrm{cntr}} + W_{\textrm{bdry}} + \cumulativeWeight{\{u_{-3},v_{-3}\}}{f}$.
    Suppose that for any $i\in\mathbb{Z}_{m+4}$, we have $\cumulativeWeight{\{u_{i+j}, v_{i+j} \mid j=0,\ldots,7\}}{f} \geqslant 8$.
    Then, the following assertions hold.
    \begin{enumerate}[label={(\roman*)}]
      \item For $r=5$, either $f$ on  $\generalizedPetersen{m+4}{1}$ has the quality-transferring property or there exists a $2\times 13$ subblock in the vicinity of $\{v_{-1},u_{-1}\}$ whose cumulative weight is at least $15$.\label{ite:cut-off-intricate-case-residue-mod-eight-is-equal-to-one-residue-mod-eight-is-equal-to-five}
      \item For $r=1$, either $f$ on $\generalizedPetersen{m+4}{1}$ has the quality-transferring property or there exists a $2\times 9$ subblock in the vicinity of $\{v_{-1},u_{-1}\}$ whose cumulative weight is at least $11$.\label{ite:cut-off-intricate-case-residue-mod-eight-is-equal-to-one-residue-mod-eight-is-equal-to-one}
    \end{enumerate}
  \end{obs}
  \begin{proof}
    \ref{ite:cut-off-intricate-case-residue-mod-eight-is-equal-to-one-residue-mod-eight-is-equal-to-five} \emph{Case 1.} $W_{\textrm{cntr}} + W_{\textrm{bdry}} \geqslant 17$.
    Then, as $f(u_{-3})+f(v_{-3})\geqslant -2$, we have that $\{u_{-3},\allowbreak{}v_{-3},u_{-2},\allowbreak{}v_{-2}\ldots,u_9,\allowbreak{}v_9\}$ is a $2\times 13$ subblock of weight at least $15$.
    
    \emph{Case 2.} $W_{\textrm{cntr}} + W_{\textrm{bdry}} = 16$.
    
    \emph{Subcase 2.1}
    $\{u_{-1}, u_{-2}, v_{-1}, v_{-2}\}\subseteq V_{-1}\cup V_{1}$ or $\{u_8,u_9,v_8,v_9\}\subseteq V_{-1}\cup V_{1}$ applies.
    
    We can just extend the $\{u_{-2},v_{-2},\ldots,u_9,v_9\}$-induced subblock to a $2\times 13$ subblock by taking into consideration the additional vertices $u_{-3},v_{-3}$ when $\{u_{-1}, u_{-2}, v_{-1}, v_{-2}\}\subseteq V_{-1}\cup V_{1}$, otherwise choosing the vertices $u_{10},v_{10}$.
    Note that this principle of extension ensures that the cumulative weight of the additionally considered vertices is necessarily at least $4$ (both vertices in $V_{-1}\cup V_1$ need to be defended by at least a $2$-labeled neighbor).
    Consequently, the weight of our considered $2\times 13$ subblock is at least $20$.
    
    \emph{Subcase 2.2} Negation of Subcase 2.1, i.e., $\{u_{-1}, u_{-2}, v_{-1}, v_{-2}\}\cap \{u_8,u_9,v_8,v_9\} \cap (V_{2}\cup V_{3})\neq\emptyset$.
    
    Upon symmetry breaking, just five label constellations meet this particular subcase: They all have in common that $\{u_{-1},v_{-1}\}\subseteq V_{-1}$ and $f(u_{9})+f(v_{9})\geqslant 3$.
    
    We focus on the $2\times 12$ subblock $B:=\{u_{-1},v_{-1},\ldots,u_{10},v_{10}\}$, which results from a right-shift\footnote{Modulo symmetry breaking we can here assume a right shift.} of the $2\times 12$ subblock $\{u_{-2},v_{-2},\ldots,u_{9},v_{9}\}$.
    By construction, the lefter-most column of $B$ consists of vertices in $V_{-1}$, while the column preceding the last column has a cumulative weight of at least $3$.
    It turns out by considering exhaustively all cases that---regardless of the labels assigned to the remaining vertices in $\{u_{0},v_{0},u_{10},v_{10}\}$---the subblock $B$ has the quality-transferring property.
    
    \emph{Case 3.} $W_{\textrm{cntr}} + W_{\textrm{bdry}} \leqslant 15$.
    
    \emph{Subcase 3.1.} $\{u_{-1}, u_{-2}, v_{-1}, v_{-2}\}\subseteq V_{-1}\cup V_{1}$ or $\{u_8,u_9,v_8,v_9\}\subseteq V_{-1}\cup V_{1}$ applies.
    
    Note that for all label-constellations for vertices in $\{u_{-1}, u_{-2}, v_{-1}, v_{-2}\} \cup \{u_8,u_9,v_8,v_9\}$ of this subcase, we have $W_{\textrm{cntr}}+W_{\textrm{bdry}}\geqslant 12$.
    Then, apparently (as in Subcase 2.1), either $u_{-2}, v_{-2}$ or $u_9,v_9$ need to be defended by $u_{-3}, v_{-3}$ respectively by $u_{10},v_{10}$, i.e., these new defending vertices must have cumulative weight at least $4$ such that we face a $2\times 13$ subblock of weight at least $16$.
    
    \emph{Subcase 3.2.} Negation of Subcase 3.1, i.e., $\{u_{-1}, u_{-2}, v_{-1}, v_{-2}\}\cap \{u_8,u_9,v_8,v_9\} \cap (V_{2}\cup V_{3})\neq\emptyset$.
    
    Exhaustively one can see that this subcase occurs only when, after symmetry breaking, one of the six constellations of labels from Table~\ref{tab:total-lower-bound-by-summation} applies.
    For each of these, the $\{u_{-3},v_{-3},\ldots,u_{9},v_{9}\}$-induced subblock of dimensions $2\times 13$ has a guaranteed lower bound of $15$; again, see Table~\ref{tab:total-lower-bound-by-summation}.\\
    \begin{table}[th]
      \centering
      \begin{longtable}{p{5.2cm}|p{1.1cm}|p{2.0cm}|p{3.45cm}|p{1.9cm}}
        {\footnotesize $\begin{bmatrix}
            f(v_{-2}) & f(v_{-1})~~~~ & f(v_{8}) & f(v_{9}) \\
            f(u_{-2}) & f(u_{-1})~~~~ & f(u_{8}) & f(u_{9})
          \end{bmatrix}$} & $W_{\textrm{bdry}}$ & LB for $W_{\textrm{cntr}}$ & LB for $\cumulativeWeight{\{u_{-3},v_{-3}\}}{f}$ & LB for $W_{t}$ \TBspacing          \\\hline\hline
        {\footnotesize $\begin{bmatrix}
            -1 & -1~~~~ & -1 & 1 \\
            1  & 3~~~~  & -1 & 2
          \end{bmatrix}$}                    & $3$                 & $10$                       & $2$                                              & $15$ \TBspacing \\\hline
        {\footnotesize $\begin{bmatrix}
            -1 & -1~~~~ & -1 & 1 \\
            2  & 3~~~~  & -1 & 2
          \end{bmatrix}$}                    & $4$                 & $10$                       & $1$                                              & $15$\TBspacing  \\\hline
        {\footnotesize $\begin{bmatrix}
            -1 & -1 ~~~~ & -1 & 1 \\
            3  & 3~~~~   & -1 & 2
          \end{bmatrix}$}                    & $5$                 & $10$                       & $0$                                              & $15$\TBspacing  \\\hline
        {\footnotesize $\begin{bmatrix}
            -1 & 1~~~~  & -1 & 1 \\
            2  & -1~~~~ & -1 & 2
          \end{bmatrix}$}                    & $2$                 & $13$                       & $3$                                              & $18$\TBspacing  \\\hline
        {\footnotesize $\begin{bmatrix}
            -1 & 1~~~~  & -1 & 2 \\
            2  & -1~~~~ & -1 & 1
          \end{bmatrix}$}                    & $2$                 & $13$                       & $3$                                              & $18$\TBspacing  \\\hline
        {\footnotesize $\begin{bmatrix}
            \hphantom{-}1 & -1~~~~ & -1 & 1 \\
            \hphantom{-}1 & 3~~~~  & -1 & 2
          \end{bmatrix}$}                 & $5$                 & $10$                       & $1$                                              & $16$\TBspacing     \\\hline
        \caption{Lower bounds (LBs) are summed up to obtain a total lower bound for $W_{t}$ (last column).\label{tab:total-lower-bound-by-summation}
        }
      \end{longtable}
    \end{table}
    
    \ref{ite:cut-off-intricate-case-residue-mod-eight-is-equal-to-one-residue-mod-eight-is-equal-to-one} \emph{Case 1.} There exists $t\in\{-1,8\}$ such that $\cumulativeWeight{\{u_t,v_t\}}{f} + \cumulativeWeight{\{u_0,v_0,\ldots,u_7,v_7\}}{f}\geqslant 11$.
    
    The $2\times 9$ subblock induced by the vertex subset $\{u_t,v_t\}\cup \{u_0,v_0,\ldots,u_7,v_7\}$ has cumulative weight at least $11$.
    
    \emph{Case 2.} For all $t\in\{-1,8\}$ we have that $\cumulativeWeight{\{u_t,v_t\}}{f} + \cumulativeWeight{\{u_0,v_0,\ldots,u_7,v_7\}}{f}< 11$.
    
    For $f$ on $\generalizedPetersen{m+4}{1}$ not having the quality-transferring property, this situation is only possible when $f$ modulo symmetry breaking satisfies
    \begin{equation}
      \begin{bmatrix}
        f(v_{-2}) & f(v_{-1})~~~~ & f(v_{8}) & f(v_{9}) \\
        f(u_{-2}) & f(u_{-1})~~~~ & f(u_{8}) & f(u_{9})
      \end{bmatrix} = \begin{bmatrix}
        1 & -1~~~~ & -1 & 1  \\
        3 & -1~~~~ & -1 & 3 
      \end{bmatrix}.\label{eq:appendix-assumption-on-f}
    \end{equation}
    
    We now show that the present scenario implies that we can spot a $2\times 12$ subblock testifying the quality-transferring property:
    Indeed, the neighboring $2\times 12$ subblock resulting from a right-shift has this property: It is induced by $\{u_{-1},v_{-1},\ldots,u_{10},v_{10}\}$ and we know for it that $\{u_{-1},v_{-1}\}\subseteq V_{-1}$ and $\{f(u_9),f(v_9)\} = \{1,3\}$.
    Finally, we note that all such  $2\times 12$ subblocks have the quality-transferring property (exhaustively, we see that the behavior is invariant from the fact how the vertices $\{u_0,v_0,u_{10},v_{10}\})$ are labeled).
  \end{proof}
  
  \begin{obs}\label{obs:small-inductive-proof-remaining-case}
    Let $L$, $R$, $C$, $C'$, $f$, $f'$ be defined as in Lemma~\ref{lem:constraing-programming-as-mathematical-lemma-formulation}.
    \begin{enumerate}[label={(\roman*)}]
      \item If $f$ furthermore satisfies the constraints $f(\nodeRTI)=f(\nodeRBI)=3$ and $f(\nodeRT)=f(\nodeRB)=-1$, then $f$ automatically guarantees that $\cumulativeWeight{C}{f}-4=\cumulativeWeight{C'}{f'}$\label{ite:the-grid-graph-scenario-which-cannot-be-proven-by-the-extending-principle}.
      \item Let $m\geqslant 9$ be odd.
      If $f$ is an optimal SDRDF for the grid graph $G_{2,m}$ with the additional property that $f(u_{m-2})=f(v_{m-2})=3$ and $f(u_{m-1})=f(v_{m-1})=-1$, then $\cumulativeWeight{G_{2,m}}{f}\geqslant m + 1$, when $m\equiv 1\pmod{4}$, otherwise, when $m\equiv 3\pmod{4}$, $\cumulativeWeight{G_{2,m}}{f}\geqslant m$.\label{ite:essential-claim-of-lemma-dealing-with-exceptional-case}
    \end{enumerate}
  \end{obs}

  \begin{proof}
    \ref{ite:the-grid-graph-scenario-which-cannot-be-proven-by-the-extending-principle} After symmetry breaking there are $129$ cases fitting these constraints.
    These all satisfy $\cumulativeWeight{C}{f}-4=\cumulativeWeight{C'}{f'}$.
    
    \ref{ite:essential-claim-of-lemma-dealing-with-exceptional-case} We show the assertions by complete induction:
    The base cases $\signedDoubleRomanDominationNumber{G_{2,9}}=10$ for $m\equiv 1\pmod{4}$ and $\signedDoubleRomanDominationNumber{G_{2,11}}=11$ for $m\equiv 3\pmod{4}$ are shown exhaustively.
    Our induction hypothesis is the claim stated in the assertion \ref{ite:essential-claim-of-lemma-dealing-with-exceptional-case}.
    For the induction step we show that $\signedDoubleRomanDominationNumber{G_{2,m+4}}\geqslant \signedDoubleRomanDominationNumber{G_{2,m}}+4$:
    Let $f$ be the function testifying $\signedDoubleRomanDominationNumber{G_{2,m+4}}=\cumulativeWeight{G_{2,m+4}}{f}$.
    On $f$, the argument from Lemma~\ref{lem:lower-bound-generalized-petersen-for-the-case-congruent-one-modulo-four} (suitable removal of vertices and addition of two edges) can be applied on the righter-most $2\times 12$ subblock $\{u_{i},v_{i}\mid i=m-12,m-11,\ldots,m-1\}$: It shows that whenever $\cumulativeWeight{G_{2,m+4}}{f}$ is strictly better than $m+4+1$ for $m\equiv 0\pmod{4}$ or better than $m+4$ for $m\equiv 3\pmod{4}$ in the assertion, it would imply the possibility to attain a strictly better bound than the proven optimum on $G_{2,m}$ (cf.~\ref{ite:the-grid-graph-scenario-which-cannot-be-proven-by-the-extending-principle})---yielding a contradiction:
    This means, for $m\equiv 1\pmod{4}$, we must necessarily have $\signedDoubleRomanDominationNumber{G_{2,m+4}}\geqslant m+4+1$, and, for $m\equiv 3\pmod{4}$, we have $\cumulativeWeight{G_{2,m+4}}{f}\geqslant m+4$.
    This concludes our inductive step.
  \end{proof}
  
  \section{Optimal labeling schemes for $G_{2,m}$}
  \begin{figure}[H]
    \centering
    \begin{minipage}{14.5cm}
      %
      %
      %
      %
      %
      %
      %
      \includegraphics{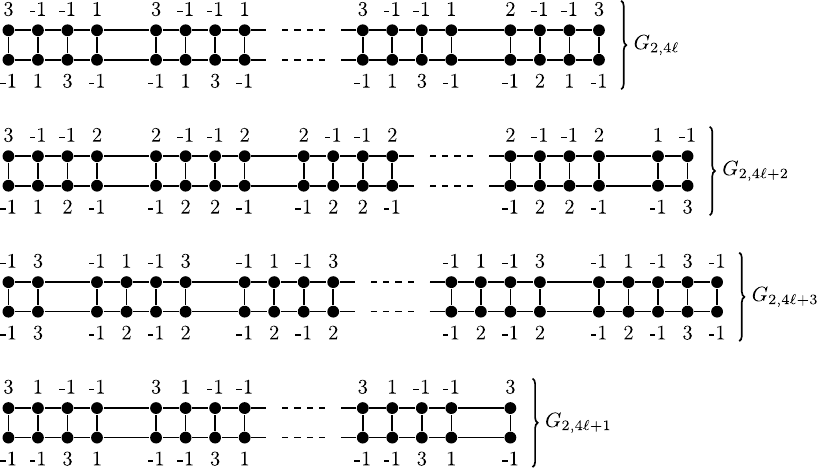}
    \end{minipage}
    \caption{Optimal labeling scheme for grid graphs depending on the congruence class of $m$ modulo $4$.
      All schemes have in common that a periodically repeating pattern of labeled $2\times 4$ grid graphs is flanked from left and/or right by differently labeled grid graphs.\label{fig:grid-two-times-m-case-congruences-three-one-mod-four}}
  \end{figure}
\end{document}